\documentclass{article}

\usepackage{arxiv}

\usepackage[utf8]{inputenc} 
\usepackage[T1]{fontenc}    
\usepackage{float}
\usepackage{hyperref}       
\usepackage{url}            
\usepackage{booktabs}       
\usepackage{amsfonts}       
\usepackage{nicefrac}       
\usepackage{microtype}      
\usepackage{lipsum}

\usepackage{mathtools}
\usepackage{algorithm2e}
\usepackage[skins, breakable]{tcolorbox}
\usepackage{subcaption}
\usepackage{mwe}
\usepackage{bbm}
\usepackage{booktabs}
\usepackage{siunitx}
\usepackage[normalem]{ulem}
\usepackage{xcolor}
\usepackage{multicol}
\usepackage{url}
\usepackage{mwe}
\usepackage{changepage}

\usepackage{amsmath}
\usepackage{amsfonts}
\usepackage{amssymb}
\usepackage{amsthm}
\usepackage{graphicx}
\usepackage{color}
\usepackage{url}
\usepackage{xspace}
\usepackage{times}
\usepackage{comment}
\usepackage{multirow}
\usepackage{booktabs}

\newcommand{\stitle}[1]{\vspace{0.5ex} \noindent{\textbf{#1}}}

\newcommand{\eop}{{\hspace*{\fill}$\Box$\par}}

\newcommand{\cequ}{Equation~}
\newcommand{\cequs}{Equations~}
\newcommand{\csec}{Section~}

\newcommand{\cthm}{Theorem~}
\newcommand{\cthms}{Theorems~}
\newcommand{\clem}{Lemma~}

\newcommand{\cprop}{Proposition~}

\newcommand{\cfig}{Figure~}

\newcommand{\capp}{Appendix~}

\newcommand{\ie}{{\em i.e.}\xspace}
\newcommand{\eg}{{\em e.g.}\xspace}
\newcommand{\etc}{{\em etc.}\xspace}
\newcommand{\etal}{{\em et al.}\xspace}

\newcommand{\squishlist}{
	\begin{list}{$\bullet$}{
		\setlength{\itemsep}{0pt}
		\setlength{\parsep}{3pt}
		\setlength{\topsep}{3pt}
		\setlength{\partopsep}{0pt}
		\setlength{\leftmargin}{1.0em}
		\setlength{\labelwidth}{1em}
		\setlength{\labelsep}{0.5em}
   }
}

\newcommand{\squishenum}{
	
	\begin{list}{\usecounter{scount}}{
		\setlength{\itemsep}{0pt}
		\setlength{\parsep}{3pt}
		\setlength{\topsep}{3pt}
		\setlength{\partopsep}{0pt}
		\setlength{\leftmargin}{1.2em}
		\setlength{\labelwidth}{1em}
		\setlength{\labelsep}{0.5em}
	}
}

\newcommand{\squishend}{
	\end{list}
}

\newcommand{\reals}{{{\mathbbm R}\xspace}}
\newcommand{\pr}[1]{{\rm Pr}\left[#1\right]\xspace}
\newcommand{\prinline}[1]{{\rm Pr}[#1]\xspace}
\newcommand{\ep}[2][]{{\rm E}_{#1}\hspace{-0.06cm}\left[#2\right]\xspace}
\newcommand{\epinline}[1]{{\rm E}[#1]\xspace}

\newcommand{\vrinline}[1]{{\rm Var}[#1]\xspace}

\newcommand{\Err}{\mathrm{Err}}
\newcommand{\indicator}[1]{{\mathbbm 1}_{#1}\xspace}
\newcommand{\bigoh}[1]{{\rm O}\!\left(#1\right)\xspace}
\newcommand{\bigohinline}[1]{{\rm O}(#1)\xspace}

\newcommand{\bigomegainline}[1]{{\rm \Omega}(#1)\xspace}

\newcommand{\bigthetainline}[1]{{\rm \Theta}(#1)\xspace}
\newcommand{\vect}[1]{{\left[ #1 \right]}\xspace}

\newcommand{\lone}{{\sf L}^1\xspace}
\newcommand{\lpnorm}{{\sf L}^p\xspace}
\newcommand{\trace}{{\rm Trace}\xspace}

\newcommand{\eps}{\epsilon\xspace}
\newcommand{\domain}{{\cal X}\xspace}
\newcommand{\domainoutput}{{\cal Y}\xspace}
\newcommand{\ldpalgo}{{\cal A}\xspace}
\newcommand{\data}{{X}\xspace}
\newcommand{\ldpdata}{\hat{X}\xspace}
\newcommand{\query}{{\bf q}\xspace}
\newcommand{\estimate}{{\bf \hat{q}}\xspace}
\newcommand{\rangeest}{{\bf \hat{\mathbf{c}}}\xspace}
\newcommand{\lap}{{\rm Lap}\xspace}
\newcommand{\diag}{{\rm diag}\xspace}

\newcommand\xicomment[1]{\textcolor{red}{Xi:#1}\xspace}

\newcommand{\ind}{{\sf ind}\xspace}
\newcommand{\ob}{\mathbf{o}\xspace}
\newcommand{\obs}{o\xspace}
\newcommand{\counts}{\mathbf{c}}
\newcommand{\rmatrix}{\mathbf{R}}

\newtheorem{lemma}{Lemma}
\newtheorem{theorem}{Theorem}
\newtheorem{definition}{Definition}

\newtheorem{proposition}{Proposition}

\PassOptionsToPackage{unicode}{hyperref}
\PassOptionsToPackage{naturalnames}{hyperref}

\begin{document}

\title{Linear and Range Counting under Metric-based Local Differential Privacy}

\author{
    Zhuolun Xiang\thanks{Work done at Alibaba Group.} \\
    UIUC\\
    \texttt{\scriptsize xiangzl@illinois.edu} \\
\And
    Bolin Ding \\
    Alibaba Group\\
    \texttt{\scriptsize bolin.ding@alibaba-inc.com} \\
\And
    Xi He \\
    University of Waterloo \\
    \texttt{\scriptsize he@uwaterloo.ca} \\
\And
    Jingren Zhou \\
    Alibaba Group \\
    \texttt{\scriptsize jingren.zhou@alibaba-inc.com} \\
}

\maketitle

\begin{abstract}
Local differential privacy (LDP) enables private data sharing and analytics without the need for a trusted data collector. Error-optimal primitives (for, \eg, estimating means and item frequencies) under LDP have been well studied. For analytical tasks such as range queries, however, the best known error bound is dependent on the domain size of private data, which is potentially prohibitive. This deficiency is inherent as LDP protects the same level of indistinguishability between any pair of private data values for each data downer.

In this paper, we utilize an extension of $\eps$-LDP called Metric-LDP or $E$-LDP, where a metric $E$ defines heterogeneous privacy guarantees for different pairs of private data values and thus provides a more flexible knob than $\eps$ does to relax LDP and tune utility-privacy trade-offs.
We show that, under such privacy relaxations, for analytical workloads such as linear counting, multi-dimensional range counting queries, and quantile queries, we can achieve significant gains in utility. In particular, for range queries under $E$-LDP where the metric $E$ is the $\lone$-distance function scaled by $\eps$, we design mechanisms with errors independent on the domain sizes; instead, their errors depend on the metric $E$, which specifies in what granularity the private data is protected.
We believe that the primitives we design for $E$-LDP will be useful in developing mechanisms for other analytical tasks, and encourage the adoption of LDP in practice.
\end{abstract}

\section{Introduction}
\label{sec:intro}
%
%
After more than a decade of research and development, differential privacy (DP) \cite{dwork2006differential} has become the {\em de facto} standard for privacy protection, and is being used or actively explored by major companies in various data applications and services, \eg, Apple \cite{url:apple}, Google \cite{erlingsson2014rappor}, Uber \cite{pvldb:johnson2017practical}, Microsoft \cite{ding2017collecting}, and Alibaba \cite{wang2019answering}. This privacy guarantee allows releasing aggregate information of the population while protecting individual's data. The degree of protection is characterized by a parameter $\eps$, which is used to tune a trade-off between the level of privacy protection and the error of data analytics.

%
Two models of DP have been studied: \emph{centralized differential privacy (CDP)} and \emph{local differential privacy (LDP)}. In CDP, a \emph{trusted} centralized data curator receives data from data owners and ensures a differentially private data release to mistrustful data analysts. In LDP, there is no trusted data curator; each data owner perturbs her data locally and sends the noisy output (LDP report) to the curator.
%

%
Recently, LDP has received a significant amount of attention in the real-world deployments of DP~\cite{erlingsson2014rappor, ding2017collecting}, as it prevents single-point failures for data breaches and relieves the burden on the data curator to keep data secure. For primitives such as frequency estimation, a sufficient number of data owners and their LDP reports (\eg, refer to lower bounds in \cite{Chan:2012:OLB:2404160.2404185, duchi2018minimax}) are required to achieve high utility. 
%
%
In more useful tasks such as range queries, more error has to be introduced with additional terms that depend on the domain size and the dimensionality. Improving the utility for queries on datasets with large domain sizes and dimensionalities, where the additional error terms are prohibitive, has been the research focus of LDP algorithms~\cite{kulkarni2019answering, highdimldp18, wang2019consistent, xu2019dpsaas}, to encourage the adoption of LDP.

In many applications, LDP is too strict and not flexible, as not all pairs of values require the same level of protection.
For instance, when website visits are collected, the website type, \eg, shopping or video website, is less sensitive than the particular website, YouTube or Hulu, or video being visited; when a person's age is collected, whether s/he is an adult or a kid is less sensitive than the exact year or month of birth.
Such relaxations have been formalized as Blowfish~\cite{he2014blowfish} and $d_{\mathcal{X}}$-privacy \cite{chatzikokolakis2013broadening} in CDP, and geo-indistinguishability \cite{andres2013geo} and Metric-LDP \cite{alvim2018local} in LDP. 
In fact, we can show that these notations are equivalent in terms of how privacy is relaxed.

In this paper, we utilize Metric-LDP \cite{alvim2018local}, which has a metric function defining different levels of privacy requirements for different pairs of values. 
We study how to make the best of such privacy relaxations to optimize utility-privacy trade-offs and to achieve provably significant utility gains for analytical tasks.
We first consider the tasks of linear counting and range counting queries under Metric-LDP.
%
%
For multi-dimensional range counting queries and a concrete class of metric, we introduce a novel mechanism whose error is {\em independent on the domain sizes of dimensions}. It achieves significantly better utility than the best known $\eps$-LDP algorithms \cite{kulkarni2019answering} and \cite{wang2019answering}, whose error is prohibitive when the domain sizes and dimensionality are non-trivial under $\eps$-LDP. Our algorithms can be applied as primitives in other tasks such as quantile queries for provable utility gains.
There were no known algorithms utilizing such relaxations to gain utility for multi-dimensional range queries in the local model. For the equivalent relaxation in CDP, the best-known utility gain \cite{haney2015design} (under Blowfish~\cite{he2014blowfish}) is much less significant than ours (relatively).

\subsection{Preliminaries}
\label{sec:pre}
Let $\domain$ denote the domain of private values. Suppose there are $n$ {\em data owners}, each holding a private value $x \in \domain$. A {\em data collector} wants to collect these private values from data owners to conduct analytical tasks.
In the {\em local model of differential privacy} (LDP), a data owner does not trust the data collector; she encodes her private value $x$ locally with a randomized algorithm $\ldpalgo$, and sends the {\em LDP report} $\ldpalgo(x)$ to the data collector.
%
%
LDP formalizes a type of plausible deniability: given any output $\ldpalgo(x)$, the likelihoods to generate $\ldpalgo(x)$ with $\ldpalgo$ from $x$ and from any other value are approximately the same.
%
%
\begin{definition}[Local Differential Privacy \cite{duchi2013local, DuchiWJ13:localsharp}]
\label{def:ldp}
A randomized algorithm $\ldpalgo: \domain \rightarrow \domainoutput$ is $\eps$-locally differentially private (or $\eps$-LDP), if for any pair of private values $x, x' \in \domain$, and any subset of output $S \subseteq \domainoutput$, we have that
$
\prinline{\ldpalgo(x)\in S}\leq e^\eps \cdot \prinline{\ldpalgo(x')\in S}.
$
\end{definition}

\stitle{Local differential privacy on metric spaces.}
$\eps$-LDP guarantees the same level of protection for all pairs of private values. However, such homogeneous privacy definition may be too strong for many applications.
We adopt an extension of LDP called {\em Metric-LDP} \cite{alvim2018local}, which uses a metric function to customize heterogeneous (different levels of) privacy guarantees for different pairs of private values and to tune utility-privacy trade-offs in analytical tasks.
\begin{definition}[Metric-based Local Differential Privacy \cite{andres2013geo, alvim2018local}]
\label{def:blowfishldp}
Let $E:\domain\times\domain \rightarrow \reals_{\geq 0}$ define a metric function on the input domain.
A randomized algorithm $\ldpalgo: \domain \rightarrow \domainoutput$ satisfies Metric-LDP or $E$-LDP if for any pair of values $x, x' \in \mathcal{X}$ and any subset of output $S\subset \domainoutput$, we have that
$
\prinline{\ldpalgo(x)\in S} \leq e^{E(x,x')} \cdot \prinline{\ldpalgo(x')\in S}.
$
\end{definition}
Here, smaller $E(x, x')$ implies that it is more sensitive to the data owners whether the private value is $x$ or $x'$. Similar to DP and LDP, $E$-LDP also has the sequential composability.

\stitle{Relationship to other relaxations.}
Metric-LDP is a generic form of Blowfish \cite{he2014blowfish} and $d_\domain$-privacy \cite{chatzikokolakis2013broadening} adapted to the local model.
In particular, Blowfish introduces the concept of {\em policy graph}, where each vertex corresponds to a data value and the distance between two vertices measures how strong the protection between the two corresponding values is (the smaller the stronger). Indeed, distance on graphs is a metric.
One attempt to further generalize the relaxation is to consider an arbitrary function $E:\domain\times\domain \rightarrow \reals_{\geq 0}$, instead of restricting $E$ to the class of metric functions.
In fact, it is sufficient to focus on the function that define a metric on $\domain$. Suppose a randomized algorithm $\ldpalgo$ is $E$-LDP. The function $E$ is said to be {\em tight} for $\ldpalgo$, if there do not exist $x,x' \in \domain$, $S \subseteq \domainoutput$, and $a < E(x,x')$, such that, $\prinline{\ldpalgo(x)\in S} \leq e^a \cdot \prinline{\ldpalgo(x')\in S}$.
It is worth noticing that for any algorithm $\ldpalgo$ that is $E$-LDP and $E$ is {\em not} tight, there exists another tight function $E'$ which $\ldpalgo$ satisfies, and $E'$ is {\em stronger} than $E$, \ie, $\forall x,x' \in \data: E'(x,x') \leq E(x,x')$ and $\exists x,x': E'(x,x') < E(x,x')$.
%
%
Therefore, it is sufficient to investigate the family of functions that are tight.
Formally, we have \cprop\ref{prop:metric_space} (with proof in \capp\ref{prop:metric_space:proof}).
\begin{proposition}\label{prop:metric_space}
Any tight function $E$ for an algorithm $\ldpalgo$ defines a metric on $\domain$.
\end{proposition}

\subsection{Problem Statement and Our Main Results}
\label{sec:definition}
Each of the $n$ data owners holds a private value $x_i \in \domain$, and let $\data = \{x_i\}_{i \in [n]}$ be the whole private dataset. 
An analytical task $\query(\data)$ is to be conducted on $\data$ by the data collector.

We focus on single-round LDP mechanisms in this work.
With $E$-LDP reports $\ldpdata = \{\ldpalgo(x_i)\}_{i \in [n]}$ collected from data owners, the data collector wants to estimate the answer to $\query(\data)$ as $\estimate(\ldpdata)$ from $\ldpdata$.
%
%
The privacy is guaranteed on each LDP report $\ldpalgo(x_i)$, and thus, we do not need to worry about privacy in designing the estimator $\estimate$ as it can be regarded as ``post-processing'' of LDP reports.
$\ldpalgo$ and $\estimate$ often need to be co-designed, as a {\em mechanism}, for an analytical task.
%

Previous work \cite{andres2013geo, alvim2018local} define utility loss as the hardness of reconstructing the real data distribution from LDP reports, represented as the expected difference between the statistical properties based on LDP reports and those based on the real data.
However, for a concrete analytical task, there is no guarantee on estimation errors for the algorithms in \cite{andres2013geo, alvim2018local}.
%
%
An important contribution of our paper is that, for several tasks, we propose mechanisms that achieve provable end-to-end utility (error bounds) under Metric-LDP.

\stitle{Linear counting (\csec\ref{sec:lcq}).} Let's consider a finite domain $\domain = [m]$.
An {\em indicator} $\indicator{{\sf P}}$ is defined to be $1$ if the predicate $\sf P$ is true, or $0$ if otherwise.
The {\em frequency vector} on the dataset $X$ is ${\bf c} = [c_x]^\intercal_{x \in [m]}$, where $c_x = \sum_{i=1}^n \indicator{x_i=x}$ represents the number owners holding a private value $x$. A {\em linear counting task} $\query(X)$ is specified by a $q \times m$ {\em workload matrix} ${\bf W}$ with $q$ rows, and asks for ${\bf W} \cdot {\bf c}$. In particular, each row of ${\bf W}$ is a {\em linear counting query} asking for a linear combination of frequencies. 
We use the {\em total expected squared error}, $\epinline{\|\estimate(\ldpdata) - {\bf W}\cdot{\bf c}\|^2}$, to measure the utility of an estimation (the expectation is taken over the randomness of $n$ instances of $\ldpalgo$).

As warm-up, for this class of counting queries, we introduce a mechanism to minimize the above error based on a generic matrix formulation, which is a reminiscence of the class of matrix mechanisms, \cite{li2010optimizing, li2012adaptive, li2013optimal} and \cite{haney2015design}, under CDP. Here, we need to carefully model the flexibility introduced by Metric-LDP to optimize the utility, \ie, allowing noises of heterogeneous magnitudes to be added at each dimension of the data. This mechanism can be applied for answering one-dim range counting queries with a provable error bound.

\stitle{Multi-dimensional range counting (\csec\ref{sec:mdrq}).}
Let's consider a $D$-dim domain $\domain = [m]^D$, and each data owner $i$ has a private value $x_i \in \domain = [m]^D$. A {\em $D$-dim range query} is specified by an interval $R = [l_1, r_1] \times \ldots \times [l_D, r_D]$, asking for $\sum_{i=1}^n \indicator{x_i \in R}$. We want to bound the {\em expected squared error} for a given range query.

A metric $E_{\lone}$ on $\domain = [m]^D$ is defined based on the $\lone$-distance: $E_{\lone}(x,y)=\eps \|x - y\|_1=\eps \sum_{i=1}^D |x[i]-y[i]|$.
%
%
%
%
For any given multi-dimensional range query on $[m]^D$, we introduce an $E_{\lone}$-LDP mechanism with expected squared error bounded by $\bigohinline{n (\frac{2}{\eps^2})^{D}}$, which completely removes the dependency on the domain size $m$ in error. Our algorithm can be extended for weighted range queries.

In comparison, the best known $\eps$-LDP algorithms \cite{kulkarni2019answering, wang2019answering} for multi-dimensional range queries have error $\bigohinline{\frac{n \log^{2D} m}{\eps^2}}$. $E_{\lone}$ is equivalent to the policy graph under Blowfish adopted by Haney \etal \cite{haney2015design} for answering range queries in the centralized setting. The techniques in \cite{haney2015design} can be extended to the local model, leading to the best previously known error bound $\bigohinline{\frac{nD (\log m)^{2(D-1)}}{\eps^2}}$ under $E_{\lone}$-LDP.

Our algorithm replaces the term $\log^{2D} m$ (in previous works) with ${1}/{\eps^{2D}}$ in the error bound.
As $\eps$ is usually chosen to be a constant no smaller than $1$ for reasonable utility in data analytics {\em under LDP}, especially in the real-world deployments, \eg, $\eps \geq 1$ in \cite{ding2017collecting} by Microsoft and $\eps \geq 4$ in \cite{url:apple} by Apple, we have ${1}/{\eps} \ll \log m$ and thus obtain a significant utility boost from the privacy relaxation.

%

%
%

\stitle{Quantile queries (\csec\ref{sec:mdrq:quantile}).}
We consider quantile queries in a one-dim domain $\domain = [m]$. 
We defer the formal definitions of quantile queries and their errors to \csec\ref{sec:mdrq:quantile}, where we will apply our algorithm for range queries as a primitive to answer quantile queries with provable accuracy gain under $E_{\lone}$-LDP.

\section{Related Work: Privacy Notations and Primitives}
\label{sec:related}

\stitle{Generalized privacy notations.}
There are several orthogonal lines of efforts that generalize or relax different aspects of the notation of differential privacy under the centralized setting. 
The first line of work is to generalize the quantification of the privacy loss, i.e., the divergence between the output distribution of an algorithm on neighboring datasets that differ in a record. Examples are  KL-\cite{WangLF2016onaverage}, Renyi-\cite{Mironov2017renyi} differential privacy, and capacity bounded differential privacy~\cite{ChaudhuriIM2019cbdp}. These generalizations aim to achieve tighter privacy composition properties than the standard differential privacy. The second line of work consider  semantic privacy frameworks which (i) clarify assumptions on the adversary and (ii) redefine sensitive information to be kept secret, such as Pufferfish privacy~\cite{kifer2012rigorous, KiferM2014pufferfish} and membership privacy~\cite{li2013membership}. Specifying a weaker version of adversary under a semantic framework~\cite{li2013membership, tramer2015differential} or weaker protection on the sensitive information~\cite{he2014blowfish, haney2015design} allow the design of algorithms with better utility than the standard differentially private algorithms. For example, Blowfish privacy~\cite{he2014blowfish} restricts the properties of an individual (from ``any pair of tuples'') that are sensitive and should not be inferred by the attacker, which are specified as a {\em policy}. \cite{he2014blowfish} shows improved utilities for several tasks including $k$-means clustering, estimating cumulative histograms, and range queries. Readers can refer to a recent survey~\cite{desfontaines2019sok} on DP for other variants under the centralized setting.

Under the local setting, MetricLDP is a direct extension of $d_{\mathcal{X}}$-privacy \cite{chatzikokolakis2013broadening, alvim2018local} from the centralized setting. It defines different privacy levels for every pair of values $x,x' \in {\cal X}$, allowing them to become more distinguishable, by a factor at most $e^{\eps \cdot d_{\cal X}(x,x')}$, as their distance $d_{\cal X}(x,x')$ increases. 
%
%
In the context of location privacy, $d_{\cal X}$ is the geographical distance, and an instance of $d_{\cal X}$-privacy and MetricLDP is called geo-indistinguishability \cite{andres2013geo, Chatzikokolakis2017ldputility}, which protects the location of the user during the interaction with location-based services.
Geo-indistinguishability (as well as $d_{\cal X}$-privacy) can be implemented via the Laplacian mechanism, and is adopted as a component in privacy-preserving mobile applications, including LP-Guardian \cite{FawazS2014ldpsmartphone}, and LP-Doctor \cite{FawazFS15ldplocation}.
A similar notion called Condensed Local Differential Privacy is proposed in \cite{gursoy2019secure}, with empirical studies on several task  such as frequency estimation, heavy hitter identification, and pattern mining.
Another related work proposes Utility-optimized LDP \cite{murakami2019utility} that protects only sensitive data with a privacy guarantee equivalent to LDP, and studies two mechanisms named utility-optimized randomized response and utility-optimized RAPPOR under such privacy definition.

In terms of the utility measurement, \cite{andres2013geo, alvim2018local} define utility loss as the hardness of reconstructing the real data distribution from LDP reports, represented as the expected difference between the statistical properties based on LDP reports and those based on the real data.
%
%
%
However, for a concrete analytical task (\eg, those in \csec\ref{sec:pre}), there is no guarantee on estimation errors for the algorithms in \cite{andres2013geo, alvim2018local}, and it is unclear how to measure the utility following the definitions in \cite{alvim2018local}.
An important contribution of our paper is that, for several tasks including linear counting queries and range queries, we propose mechanisms that achieve provable end-to-end utility gain under reasonable metric functions.


\stitle{Shuffle model.}
Orthogonal to the relaxation of privacy notations is a recent line of work~\cite{prochlo17, mixnets18, amplification18, Balle2019blanket} connecting LDP to the centralized model by patching an $\eps_0$-LDP algorithm with a trusted shuffler who randomly shuffles LDP reports from $n$ data owners (\eg, using an anonymous communication channel). If $\eps_0$ is small, the shuffling step amplifies the privacy guarantee to be $(\eps, \delta)$-DP, where $\eps = \bigohinline{\eps_0/\sqrt{n}}$. Such {privacy amplification} implies {accuracy boost}, because for a reasonable privacy budget, \eg, $\eps = 2$, it allows LDP algorithms to use a larger $\eps_0$ to achieve better accuracy.
%
%

\stitle{Primitives under LDP.}
Finally, we give a brief summary on the analytical primitives under LDP (without relaxation). Mean/median estimation under $\eps$-LDP has been well studied \cite{duchi2013local, duchi2018minimax, ding2017collecting} with a matching upper and lower bound.
Frequency estimation under LDP is also studied extensively in, \eg, \cite{DuchiWJ13:localsharp, erlingsson2014rappor, stoc:BassilyS15, ding2017collecting, nips:BassilyNST17, wang2017locally, aistats:AcharyaSZ19}. They use techniques like hashing (\eg, \cite{wang2017locally}) and Hadamard transform (\eg, \cite{nips:BassilyNST17, aistats:AcharyaSZ19}) for good utility.
%
%
%
Consistent frequency estimation which requires the estimations to be non-negatives and sum up to $1$ is investigated in a recent work~\cite{wang2019consistent}.
%
%
For locally differentially private range queries, the work of \cite{kulkarni2019answering, wang2019answering, xu2019dpsaas} present the state-of-the-art.
%
%
One goal of this paper is to achieve better utility in such analytical primitives via relaxation and tuning the knob of metric function.

\section{Frequency Estimation and Linear Counting Queries}
\label{sec:lcq}
We first consider the task of answering linear counting queries, defined in \csec\ref{sec:definition}: how to collect each private value in $X$ under $E$-LDP, and estimate ${\bf W} \cdot {\bf c}$ for a given workload matrix $\bf W$.



A straightforward framework for answering linear counting queries is to first estimate the frequency vector as $\hat{\bf c}$ from $E$-LDP reports $\ldpdata = \{\ldpalgo(x_i)\}_{i \in [n]}$, and return ${\bf W} \cdot \hat{\bf c}$. Assuming the estimation is unbiased, \ie, $\epinline{\hat{\bf c}} = {\bf c}$, with the goal of minimizing the total expected squared error, we want to minimize $\epinline{\|\estimate(\ldpdata) - {\bf W}\cdot{\bf c}\|^2}=\trace({\bf W}^\intercal {\bf W} \cdot \vrinline{\hat{\bf c}})$.
There are two remaining questions: i) how to estimate the frequency vector under $E$-LDP; and ii) how the metric function $E$ can help gain utility in this framework. We give a matrix formulation next to answer these two questions.

\subsection{A Generic Matrix Formulation under Metric-LDP}
The class of {\em matrix mechanism} has been well studied in the centralized setting of differential privacy, to reduce the variance for estimating a workload of linear counting queries \cite{li2010optimizing, li2012adaptive, li2013optimal}. Haney \etal \cite{haney2015design} establish an equivalence relationship between Blowfish privacy for answering linear counting queries and standard $\eps$-differential privacy for answering transformed linear counting queries, under some policies. We show that under $E$-LDP, a similar mechanism can be formulated.

The intuition behind matrix mechanisms in the centralized setting \cite{li2010optimizing, haney2015design} is that, instead of answering a workload $\mathbf{W}$ with high sensitivity, the data collector can answer a properly chosen {\em strategy workload} $\mathbf{A}$ with low sensitivity, and then reconstruct the answers for the workload $\mathbf{W}$.
For the local setting, we can generalize the above idea by asking each data owner to prepare the $E$-LDP report according to the strategy workload $\mathbf{A}$ instead of the actual $\mathbf{W}$. Then the data collector reconstruct the answers for $\mathbf{W}$ via linear transformations. A properly chosen strategy matrix may lower the amount of noise to be injected to the $E$-LDP reports and thus improve the utility. 

\stitle{$E$-LDP encoding algorithm $\ldpalgo_{{\bf A}, {\bf B}, {\bf s}}(x)$.}
More formally, every data owner uses the same $p \times m$ strategy matrix $\mathbf{A}=[\mathbf{a}_1 \ldots \mathbf{a}_p]^\intercal$. Every row of the workload matrix $\mathbf{W}$ can be reconstructed using a linear combination of rows of $\mathbf{A}$. That is, there exists matrix decomposition $\mathbf{W}=\mathbf{B}\mathbf{A}$ for some $q \times p$ matrix $\mathbf{B}$.
Each data owner first encode her value $x$ as a length-$m$ binary vector $\mathbf{h}_x=[0,...,0,1,0,...,0]^\intercal$ where only the $x$-th position is $1$.
We use $\lap(s)$ to represent a random sample drawn from Laplace distribution with parameter $s$. Each data owner draws $p$ independent random samples $\lap({\bf s})=\vect{\lap(s_1), \ldots, \lap(s_p)}^\intercal$, with parameters ${\bf s} = [s_1, \ldots, s_p]^\intercal$, and reports:
\[
\ldpalgo_{{\bf A}, {\bf B}, {\bf s}}(x) = \mathbf{A} \cdot \mathbf{h}_x + [\lap(s_1), \ldots, \lap(s_p)]^\intercal = \mathbf{A} \cdot \mathbf{h}_x + \lap({\bf s}).
\]
\begin{proposition}
$\ldpalgo_{{\bf A}, {\bf B}, {\bf s}}$ is $E$-LDP, if for any pair of $x, x' \in [m]$, we have
\[
    \begin{bmatrix}
        \frac{1}{s_1} & \frac{1}{s_2} & \cdots & \frac{1}{s_p}
    \end{bmatrix}
    |\mathbf{A}(\mathbf{h}_x-\mathbf{h}_{x'})| \leq E(x,x'),
\]
where $|\mathbf{A}(\mathbf{h}_x-\mathbf{h}_{x'})|=
\begin{bmatrix}
|\mathbf{a}^\intercal_1 (\mathbf{h}_x-\mathbf{h}_{x'})| & \cdots & |\mathbf{a}^\intercal_p (\mathbf{h}_x-\mathbf{h}_{x'})|
\end{bmatrix}^\intercal$,
namely, $|\mathbf{A}(\mathbf{h}_x-\mathbf{h}_{x'})|$ is the vector obtained by taking the absolute values of entries in vector $\mathbf{A}(\mathbf{h}_x - \mathbf{h}_{x'})$.
\end{proposition}
\begin{proof}
To show $\ldpalgo_{{\bf A}, {\bf B}, {\bf s}}$ satisfies $E$-LDP, we just need to show
\begin{align*}
\frac{\prinline{\ldpalgo_{{\bf A}, {\bf B}, {\bf s}}(x) = \mathbf{r}}}{\pr{\ldpalgo_{{\bf A}, {\bf B}, {\bf s}}(x')=\mathbf{r}}} & = \prod_{i=1}^p\frac{\exp(-\frac{|\mathbf{r}[i]-\mathbf{A}\cdot\mathbf{h}_x[i]|}{s_i})}{\exp(-\frac{|\mathbf{r}[i]-\mathbf{A}\cdot\mathbf{h}_{x'}[i]|}{s_i})} \leq
\prod_{i=1}^p \exp(\frac{|\mathbf{A}(\mathbf{h}_x[i]-\mathbf{h}_{x'}[i])|}{s_i})
\\
& = \exp\left(\sum_{i=1}^p\frac{|\mathbf{A}(\mathbf{h}_x[i]-\mathbf{h}_{x'}[i])|}{s_i}\right) \leq e^{E(x,x')},
\end{align*}
from the condition in the proposition.
\end{proof}

\stitle{Answering linear counting workload.}
After collecting $\ldpdata = \{{\bf r}_i = \ldpalgo_{{\bf A}, {\bf B}, {\bf s}}(x_i)\}_{i \in [n]}$ from $n$ data owners, the data collector estimates the linear counting queries ${\bf W} \cdot {\bf c}$ as $\mathbf{B}\cdot\sum_{i=1}^n \mathbf{r}_i$.
\begin{proposition}\label{prop:mmerror}
The estimation $\estimate(\ldpdata) = \mathbf{B} \cdot \sum_{i=1}^n \mathbf{r}_i$ is an unbiased estimation of ${\bf W} \cdot {\bf c}$. The variance (total expected squared error) of the estimation $\estimate(\ldpdata)$ is 
\[
 \epinline{\|\estimate(\ldpdata) - {\bf W}\cdot{\bf c}\|^2}
= 2n \cdot \trace[\mathbf{B}^\intercal\mathbf{B} \cdot \diag(s_1^2, \ldots, s_p^2)]
\]
where $\diag(s_1^2, \ldots, s_p^2)$ is a $p\times p$ diagonal matrix with diagonal elements $s_1^2, \ldots, s_p^2$.
\end{proposition}
The proof of this proposition is in \capp\ref{prop:mmerror:proof}.
%
    

%
If we set ${\bf A} = {\bf I}_m$ ($m \times m$ identity matrix) and ${\bf B} = {\bf W}$, then $\sum_{i=1}^n \mathbf{r}_i$ is an unbiased estimation of the frequency vector ${\bf c}$. In general, we can choose ${\bf A}$ and ${\bf B}$ properly based on the workload $\bf W$ and the metric $E$ to gain utility, which is formulated as the following optimization problem.

\stitle{An optimization problem.}
Given a workload $\bf W$ and a metric $E$, we want to choose $\bf A$, $\bf B$, and ${\bf s} = [s_1, \ldots, s_p]^\intercal$ to minimize the total expected squared error on the workload $\bf W$:
\begin{equation}\label{opt:matrix_mechanism}
\begin{aligned}
\min_{{\bf A}, {\bf B}, {\bf s}} \quad & 2n \cdot \trace[\mathbf{B}^\intercal \mathbf{B} \cdot \diag(s_1^2, \ldots, s_p^2)]
\\
\hbox{s.t.~} &
\begin{bmatrix}
    \frac{1}{s_1} & \frac{1}{s_2} & \cdots & \frac{1}{s_p}
\end{bmatrix}
|\mathbf{A}(\mathbf{h}_x-\mathbf{h}_{x'})| \leq E(x,x'), \quad \forall x, x' \in [m]
\\
& \mathbf{B}\mathbf{A}=\mathbf{W}
\\ 
& s_k>0, \quad \forall k \in [p]
\end{aligned}
\end{equation}
where
$|\mathbf{A}(\mathbf{h}_x-\mathbf{h}_{x'})|=
\begin{bmatrix}
    |\mathbf{a}^\intercal_1 (\mathbf{h}_x-\mathbf{h}_{x'})| & \cdots & |\mathbf{a}^\intercal_p (\mathbf{h}_x-\mathbf{h}_{x'})|
\end{bmatrix}^\intercal$.




It is hard to solve \eqref{opt:matrix_mechanism} efficiently, unless $\bf A$ is fixed (then it becomes convex but a bad choice of $\bf A$ may lead to a suboptimal solution).  In the rest of this section, we will consider several specific workload matrices $\mathbf{W}$ and metric functions $E$ that have interesting usage in practice and admit efficient solutions.


\subsection{Minimizing Total Error in Frequency Queries}
\label{sec:lcq:fq}
%
%
When the workload matrix ${\bf W} = {\bf I}_m$ (an $m \times m$ identity matrix), the problem becomes estimating the frequencies of all private values. We choose both $\bf A$ and $\bf B$ in \eqref{opt:matrix_mechanism} to be ${\bf I}_m$. The $E$-LDP encoding algorithm $\ldpalgo_{{\bf A}, {\bf B}, {\bf s}}(x)$ introduced above, parameterized by $\bf s$, becomes a metric-based extension of the {\em histogram encoding} mechanism introduced in \cite{wang2017locally} for constructing the frequency oracle under $\eps$-LDP. And accordingly, the optimization problem \eqref{opt:matrix_mechanism} becomes a convex and solvable special case:
\begin{equation}\label{opt:matrix_mechanism:fq}
\begin{aligned}
\min_{\bf s} \quad & 2n \cdot \sum_{x=1}^m s_x^2
\\
\hbox{s.t.~} & \frac{1}{s_x}+\frac{1}{s_{x'}}\leq E(x, x'), \quad \forall x, x' \in [m]
\\
& s_x > 0, \quad \forall x \in [m].
\end{aligned}
\end{equation}

We can achieve better utility under $E$-LDP than standard $\eps$-LDP by solving the optimization problem \eqref{opt:matrix_mechanism:fq} for non-trivial metric functions. Let's look at such a family of metric functions $E_S$, parameterized by $S \subseteq [m]$, which denotes the set of {\em super sensitive values} $[m]$ (while values $[m] - S$ are less sensitive). More formally, let $E_S(x, x') = \eps$ if $x \in S$ or $x' \in S$, and $E_S(x, x') = 2\eps$ otherwise. Note that $E_S$ is a metric space by definition.
Figure \ref{fig:star_policy}(a) illustrates the above metric function, where $S$ consists all red nodes. All edges connected to those sensitive values are $\eps$ and rest of the edges are $2\eps$.
For example, in smart building, this metric function means that some locations $S$ are more sensitive (\eg, restroom, smoker lounge, \etc) than the other locations (\eg, meeting room). Offering stronger privacy guarantee to more sensitive locations can be specified using the metric $E_S$ above.


We solve the above optimization problem under $E_S$, and compare the resulting total expected squared error with the one of a standard $\eps$-LDP frequency estimation mechanism \cite{wang2017locally} which does not take care of the heterogeneous privacy requirements for different pairs of private values and has fixed expected squared error.
The numeric results are plotted in \cfig\ref{fig:star_policy} with $m = 100$ and $|S|$ varies from $1$ to $100$. It can be seen that when the number of super sensitive values is smaller (\eg, $|S| \leq 10$), $E_S$-LDP gives better utility; when $|S|$ reaches $40$, $E_S$-LDP almost has the same utility as $\eps$-LDP.


\begin{figure}[ht]
    \centering
    \begin{subfigure}[c]{0.4\textwidth}
        \centering
        \includegraphics[height=3cm]{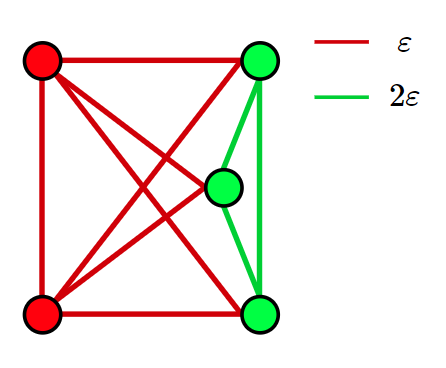}
        \caption{Illustration of metric $E_S$}
    \end{subfigure}
    ~
    \begin{subfigure}[c]{0.58\textwidth}
        \centering
        \includegraphics[height=3cm]{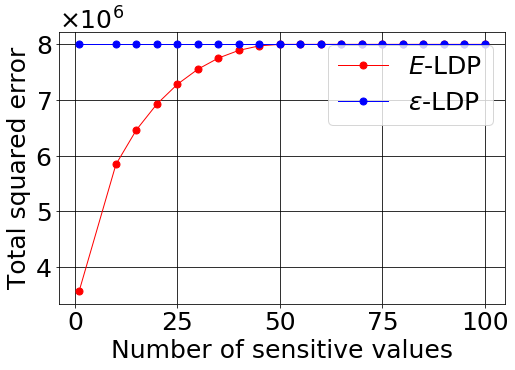}
        \caption{Varying $|S|$, $m=100$, $n=10^4$, $\eps=1.0$}
    \end{subfigure}
    
    \caption{Comparison between utility under $E_S$-LDP and under $\eps$-LDP for frequency queries}
    \label{fig:star_policy}
\end{figure}

\subsection{One-dimensional Range Queries}
\label{sec:lcq:1drq}
We consider one-dim range queries now. A range query is specified by an interval $R = [l, r] \subseteq [m]$, and asks for $\sum_{i = 1}^n \indicator{x_i \in R}$. When the distance between private values matters, it is natural to consider a metric $E$: $E(x, x') = \eps|x - x'|$ for $x, x' \in [m]$, which means values that are closer are more sensitive to each other (it is equivalent to the local version of line graph policy in Blowfish \cite{he2014blowfish} and MetricLDP \cite{alvim2018local}). For example, when $x$ is the age of a person, it is fine to release the information whether a person is an adult or a kid, but the exact year or month of birth is more sensitive.
Mechanisms are developed to handle range queries under $\eps$-LDP \cite{kulkarni2019answering, wang2019answering} (more will be discussed in Section \ref{sec:mdrq}). The hope is that we can achieve better utility under $E$-LDP by solving the problem \eqref{opt:matrix_mechanism}.

Let ${\bf W}_m$ be the workload matrix for all possible one-dimensional range queries on $[m]$.
%
%
We consider a strategy matrix $\mathbf{A}=\mathbf{L}_m$ (an $m\times m$ $\{0,1\}$-matrix with bottom-left triangular area filled with $1$), which intuitively means that each user creates an LDP report for estimating every prefix sum of the frequencies (a range query can be answered as the difference between two prefixes).
\begin{equation}
\mathbf{W}_3=
\begin{tiny}
\begin{bmatrix}
    1 & 0 & 0\\
    1 & 1 & 0\\
    1 & 1 & 1\\
    0 & 1 & 0\\
    0 & 1 & 1\\
    0 & 0 & 1
\end{bmatrix}
\end{tiny}
\quad
\mathbf{L}_3=
\begin{small}
\begin{bmatrix}
    1 & 0 & 0\\
    1 & 1 & 0\\
    1 & 1 & 1
\end{bmatrix}
\end{small}
\end{equation}
Above are examples of ${\bf W}_3$ and ${\bf L}_3$. With ${\bf A} = {\bf L}_m$, we can solve the problem \eqref{opt:matrix_mechanism} as:
\begin{equation*}
\begin{aligned}
    \min_{\bf s} \quad & 2n \sum_{x=1}^m m s_x^2 \\
    \hbox{s.t.~} & \sum_{i = l}^{r-1} \frac{1}{s_{i}} \leq \eps (r-l), ~\forall 1 \leq l < r \leq m \\
    & s_x \geq 0, ~\forall x \in [m]
\end{aligned}.
\end{equation*}
We can easily derive the optimal solution to the above problem as $s_k=\frac{1}{\eps}$ for $k \in [m-1]$ and $s_m = 0$, and thus, the total expected squared error is $\frac{2nm(m-1)}{\eps^2}$ for $m(m-1)/2$ range queries. For each range query, the squared error is $\bigohinline{\frac{n}{\eps^2}}$. This already gives an $\bigohinline{\log^2 m}$ improvement (indeed, under a relaxed privacy notation, $E$-LDP) on the utility in comparison to the mechanisms in \cite{kulkarni2019answering} and \cite{wang2019answering} (which have expected squared error $\bigohinline{\frac{n\log^2 m}{\eps^2}}$, under $\eps$-LDP).

In the next section, we will continue investigating this important workload class, range queries, with two goals: i) bounding the error per query (instead of the total error); and ii) handling $D$-dimensional range queries on $[m]^D$ with error bounds independent on the domain size $m$.


\section{Multi-dimensional Range Query and Quantile Search}
\label{sec:mdrq}
We now consider the task of answering range counting queries in a $D$-dim domain, defined in \csec\ref{sec:definition}: a {range query} is specified by a range $R = [l_1, r_1] \times \ldots \times [l_D, r_D] \subseteq [m]^D$, asking for $c(R) = \sum_{i=1}^n \indicator{x_i \in R}$. We provide $E_{\lone}$-LDP when collecting each $x_i$.


Our results can be extended for $E_{\lpnorm}$, due to the relation: $\|x\|_p \leq \|x\|_1 \leq D^{1-\frac{1}{p}} \|x\|_p$ for any $p\geq 1$.
Therefore, any algorithm that is $E_{\lone}$-LDP with parameter $\eps$ is also $E_{\lpnorm}$-LDP with parameter $D^{1-\frac{1}{p}}\eps$.

\stitle{Notations.} We say that a dimension $d$ of the range query $R=[l_1,r_1]\times\cdots\times[l_D,r_D]$ is {\em trivial} if $[l_d,r_d]=[1,m]$, and {\em nontrivial} otherwise. Let $D_R$ be the number of nontrivial dimensions of $R$.
%

For a value $x \in [m]^D$ or any vector ${\bf v}$, we use $x[i]$ or ${\bf v}[i]$ to denote the coordinate on the $i$th dimension, respectively.
%
%
We assign an {\em index} ($\ind: [m]^D \rightarrow [m^D]$) to each value in the $D$-dim domain $[m]^D$, numbering all the values in $[m]^D$ from $1$ to $m^D$: $\ind(x)=1+\sum_{d=1}^D m^{d-1}(x[d] - 1)$. If it is clear from the context, we will refer to $x$ as both a value in $[m]^D$ and its index $\ind(x)$, interchangeably.

\stitle{Comparison to existing approaches and our main results.}
Existing methods for answering range queries under $\eps$-LDP are either based on hierarchical histograms \cite{kulkarni2019answering, wang2019answering} or discrete Haar transform \cite{kulkarni2019answering}. For approaches based on hierarchical histograms, a one-dim range query can be split into $\bigohinline{\log m}$ sub-queries in a hierarchy of intervals, and each sub-query can be answered via frequency estimation under $\eps$-LDP. The expected squared error is thus $\bigohinline{\frac{n \log^2 m}{\eps^2}}$ (when $\eps$ is small). The same idea can be extended for $D$-dim domain and queries with $D_R$ nontrivial dimensions, with error $\bigohinline{\frac{n \log^{D_R+D} m}{\eps^2}}$. Approaches based on discrete Haar transform have the same asymptotic error.

Our goal here is to remove the prohibitive $(\log m)^{\bigohinline{D}}$ term from error bounds under $E_{\lone}$-LDP. Both schemes introduced above, however, rely on $(\log m)^{\bigohinline{D}}$ independent frequency estimations per query, and each frequency estimation as a black box is inherently hard with error as least $\bigomegainline{\frac{n}{\eps^2}}$ even under $E_{\lone}$-LDP (consider a domain with two possible values). And thus, {\em the $(\log m)^{\bigohinline{D}}$ term is inevitable for existing methods even under relaxation}.
In the centralized model (Blowfish privay), under the same metric $E_{\lone}$, Haney \etal \cite{haney2015design} made some improvement but failed to completely remove the $(\log m)^{\bigohinline{D}}$ term. Their approach has expected squared error $\bigohinline{\frac{D (\log m)^{3(D-1)}}{\eps^2}}$ under $E_{\lone}$-CDP, which is only better than the Privelet mechanism \cite{XiaoWG2010wavelet} under $\eps$-CDP by a $\bigthetainline{\log^{3} m}$ factor.
Haney \etal \cite{haney2015design}'s method can be extended to the local model $E_{\lone}$-LDP, with error $\bigohinline{\frac{nD (\log m)^{2(D-1)}}{\eps^2}}$, reducing the expected squared error in the methods \cite{kulkarni2019answering, wang2019answering} under $\eps$-LDP only by a factor of $\bigthetainline{\log^{2} m}$.

We propose an efficient mechanism with error bounded by $\bigohinline{n (\frac{2}{\eps^2})^{D}}$ when $\eps$ is small, independent on the domain size $m$ (\csec\ref{sec:mdrq:alg}). As $\eps$ is usually chosen to be a constant (\eg, $1$), $\frac{1}{\eps} \ll \log m$, and thus replacing $\log^{2D} m$ with $\frac{1}{\eps^{2D}}$ improves the utility significantly.
Our method can be considered as a special type of transformation similar to discrete Haar transform, but with a nice property that during the summation of frequency estimations of single values, most noise from perturbation will be canceled out.
%
%
Our method naturally extends to the case where each dimension has a different size (\csec\ref{sec:mdrq:domain}), \ie, $\domain = [m_1] \times \cdots \times[m_D]$.
When each data owner $i$ holds a (private) weight $w_i\in W$, a weighted range query asks $c_{\bf w}(R)=\sum_{i=1}^n w_i \indicator{x_i\in R}$, which can be also handled by our method (\csec\ref{sec:mdrq:weighted}),
%
%
with error $\bigohinline{{n \Delta^2 (\frac{2}{\eps^2})^{D+1}}}$, where $\Delta=\max_{w_i\in W}{|w_i|}$.
Finally, we introduce how to apply our method to find quantiles under $E_{\lone}$-LDP, and analyze the error (\csec\ref{sec:mdrq:quantile}).
It is likely that the techniques we develop can be applied in the centralized setting \cite{haney2015design} (Blowfish) to completely remove the $(\log m)^{\bigohinline{D}}$ error term there, but we will leave it as future work.

\subsection{Multi-dimensional Range Query under E-LDP}
\label{sec:mdrq:alg}
\stitle{$E_{\lone}$-LDP encoding algorithm $\ldpalgo(x)$.}
Let $x$ denote the $D$-dimensional private value held by a data owner. She first encodes each dimension $d$ of $x$, $x[d]\in[m]$, into a length-$m$ vector $\mathbf{b}_d\in \{-1,1\}^m$:
\[
\mathbf{b}_d=[\underbrace{-1,-1,\ldots,-1}_{x[d]-1},\underbrace{1,1,\ldots,1}_{m-x[d]+1}]
\]
where the first up to the $(x[d]-1)$-th position are $-1$'s and the rest are $1$'s.
She will then perturb the vector $\mathbf{b}_d$ into $\mathbf{r}_d$ with standard random-flipping operation on each position $k\in [m]$:
\[
\mathbf{r}_d[k]= \begin{dcases*}
\mathbf{b}_d[k] & with prob. $\frac{e^\eps}{e^\eps+1}$\\
-\mathbf{b}_d[k] & with prob. $\frac{1}{e^\eps+1}$
\end{dcases*}.
\]
The data owner reports the $D \times m$ matrix $\ldpalgo(x)=\mathbf{R} =\vect{\mathbf{r}_1,\mathbf{r}_2,...,\mathbf{r}_D}^\intercal$ to the data collector.
\begin{proposition}
    $\ldpalgo(x)$ is $E_{\lone}$-LDP.
\end{proposition}
\begin{proof}
    To show it satisfies $E_{\lone}$-LDP, we have the following for any $x,x'$:
    \[
    \frac{\Pr[\ldpalgo(x) = \mathbf{R}|x]}{\Pr[\ldpalgo(x') = \mathbf{R}|x']}\leq \prod_{i=1}^D e^{\eps|x[i]-x'[i]|} = e^{\eps\sum_{i=1}^D|x[i]-x'[i]|}=e^{E(x,x')}.
    \]
\end{proof}

\stitle{Range query estimation.}
After collecting data owners' reports $\mathbf{R}_1,...,\mathbf{R}_n$, where $\mathbf{R}_i = \ldpalgo(x_i)$, the data collector first obtains a length-$m^D$ vector $\ob = [\obs_1,\ldots,\obs_{m^D}]^\intercal$, called {\em observations}:
\begin{equation}\label{eq:observation}
    \obs_{x} = \sum_{i=1}^n \prod_{d=1}^D \mathbf{R}_i[d, x[d]], \quad \forall x\in[m]^D
\end{equation}
where $\mathbf{R}[a,b]$ denotes the value in row $a$ and column $b$ of matrix $\mathbf{R}$. 

Recall that the {index} $\ind: [m]^D \rightarrow [m^D]$ numbers all the values in $[m]^D$ from $1$ to $m^D$, namely, $\ind(x)=1+\sum_{d=1}^D m^{d-1}(x[d] - 1)$. When referring to indexes of entries in a vector, we will use $x$ and $\ind(x)$, interchangeably.
Thus, by $\obs_x$, we mean the $\ind(x)$-th position $o_{\ind(x)}$ in the vector $\ob$.

For example, if $n=2$, $D=2$ and $m=3$, with $\mathbf{R}_1=
\begin{tiny}
\begin{bmatrix*}[r]
    1 & -1 & 1\\
    -1 & -1 & -1\\
\end{bmatrix*}
\end{tiny}$ and
$\mathbf{R}_2=
\begin{tiny}
\begin{bmatrix*}[r]
    1 & 1 & -1\\
    1 & -1 & -1\\
\end{bmatrix*}
\end{tiny}$, for $x=(1,1)$, 
$\obs_x=\mathbf{R}_1[1, x[1]]\cdot \mathbf{R}_1[2, x[2]]+\mathbf{R}_2[1, x[1]]\cdot \mathbf{R}_2[2, x[2]]=1\cdot (-1) + 1\cdot 1=0$.

We will use $\ob$ to estimate the frequencies of all values in $[m]^D$. Let $\mathbf{c}=[c_1,...,c_{m^D}]^\intercal$ be the vector representing true frequencies of all values $x\in [m]^D$ among the $n$ data owners.
%
%

As will be proved in Theorem \ref{thm:mdrq_correctness}, there exists a relation 
\begin{equation} \label{eq:freq_est:relation}
\epinline{\ob} =(\frac{e^{\eps}-1}{e^{\eps}+1})^D\mathbf{B}_{m,D}\cdot \mathbf{c},    
\end{equation}
where $\mathbf{B}_{m,D}$ is an $m^D\times m^D$ matrix that can be partitioned into $m\times m$ submatrices $\mathbf{B}_{m,D-1}$, satisfying the following recursive relation for $2\leq d \leq D$,
\begin{equation}\label{eq:B_recursion}
\mathbf{B}_{m,d}=
\begin{bmatrix}
\mathbf{B}_{m,d-1} & -\mathbf{B}_{m,d-1} & \cdots  & -\mathbf{B}_{m,d-1}\\
\vdots & \ddots & \ddots & \vdots\\
\vdots &  & \ddots & -\mathbf{B}_{m,d-1}\\
\mathbf{B}_{m,d-1} & \cdots & \cdots & \mathbf{B}_{m,d-1}\\
\end{bmatrix}
\end{equation}
That is, after partition, the submatrices in the bottom-left triangle are all $\mathbf{B}_{m,d-1}$ and rest of the submatrices are all $-\mathbf{B}_{m,d-1}$. For the base case when $D=1$,
\begin{equation*}
\mathbf{B}_{m,1}=
\begin{bmatrix}
1 & -1 & \cdots  & -1\\
\vdots & \ddots & \ddots & \vdots\\
\vdots &  & \ddots & -1\\
1 & \cdots & \cdots & 1\\
\end{bmatrix}
\end{equation*}

\textbf{Estimate Single-value Frequencies.}
The estimated frequency vector 
$\rangeest=[\hat{c}_1,...,\hat{c}_{m^D}]^\intercal$ can be thus computed from \eqref{eq:freq_est:relation}  as follows:
\begin{equation}\label{eq:freq_est}
    \rangeest=
    (\frac{e^\eps+1}{e^\eps-1})^{D}\mathbf{B}^{-1}_{m,D}\cdot\ob,
\end{equation}

For any value $x\in [m]^D$, $\hat{c}_x$ is the estimated frequency of $x$.

\textbf{Estimate answers to range queries.}
For any value $x\in [m]^D$, $\hat{c}_x=\hat{c}_x$ is the frequency of $x$ estimated as \eqref{eq:freq_est}.
For a $D$-dim range query $R=[l_1,r_1] \times \cdots \times [l_D,r_D]$, the data collector can estimate its answer by directly summing up the estimated frequencies of all $x\in R$, that is,
\begin{equation}\label{eq:range_freq_est}
    \hat{c}(R)=\sum_{x\in R}\hat{c}_x =
    (\frac{e^\eps+1}{e^\eps-1})^{D}\sum_{x\in R}\mathbf{e}_x \mathbf{B}^{-1}_{m,D}\cdot\ob,
\end{equation}
where $\mathbf{e}_x$ is a $0$-$1$ row vector with only the $\ind(x)$-th entry as $1$, and $\mathbf{e}_x \mathbf{B}^{-1}_{m,D}$ gives the $\ind(x)$-th row in $\mathbf{B}^{-1}_{m,D}$.

\textbf{Computing $\mathbf{B}^{-1}_{m,D}$.}
The rest question is thus how to compute the matrix inverse $\mathbf{B}^{-1}_{m,D}$.
%
%
It turns out that we can efficiently compute it in a recursive way. 
$\mathbf{B}^{-1}_{m,D}$ can be partitioned into $m\times m$ submatrices $\mathbf{B}^{-1}_{m,D-1}$, defined by the following recursive relation for $2\leq d \leq D$:
\begin{equation}\label{eq:recursion}
\mathbf{B}_{m,d}^{-1}=
\frac{1}{2}
\begin{bmatrix}
\mathbf{B}_{m,d-1}^{-1} & 0 & \cdots  & 0 & \mathbf{B}_{m,d-1}^{-1}\\
-\mathbf{B}_{m,d-1}^{-1} & \mathbf{B}_{m,d-1}^{-1} & \ddots & \vdots & 0\\
0 & -\mathbf{B}_{m,d-1}^{-1} & \ddots & 0 & \vdots\\
\vdots & \ddots & \ddots & \mathbf{B}_{m,d-1}^{-1} & 0\\
0 & \dots & 0 & -\mathbf{B}_{m,d-1}^{-1} & \mathbf{B}_{m,d-1}^{-1}\\
\end{bmatrix}.
\end{equation}
Recursively, $\mathbf{B}_{m,d-1}^{-1}$ is a $m^{d-1}\times m^{d-1}$ matrix. In the base case, $\mathbf{B}_{m,1}^{-1}$ is the $m\times m$ matrix:
\[
\mathbf{B}_{m,1}^{-1}=
\frac{1}{2}
\begin{bmatrix}
1 & 0 & \cdots  & 0 & 1\\
-1 & 1 & \ddots & \vdots & 0\\
0 & -1 & \ddots & 0 & \vdots\\
\vdots & \ddots & \ddots & 1 & 0\\
0 & \dots & 0 & -1 & 1\\
\end{bmatrix}.
\]

We will show that the estimated answer $\hat{c}(R)$ is unbiased and bound its error in \csec\ref{sec:mdrq:analysis}. We will also present simulation results on accuracy and analyze the complexity of our solution.

\subsection{Analysis of Algorithm}
\label{sec:mdrq:analysis}
\stitle{Accuracy analysis.}\label{sec:mdrq:err}
We first show the unbiasedness of our estimations in \csec\ref{sec:mdrq:alg}.
\begin{theorem}\label{thm:mdrq_correctness}
    The estimates for the frequency of any single value and the answer to any range query $R$ (\cequs\eqref{eq:freq_est} and (\ref{eq:range_freq_est}), respectively) are unbiased, \ie, $\epinline{\hat{\mathbf{c}}}=\mathbf{c}$ and  $\epinline{\hat{c}(R)}=c(R)$.
\end{theorem}
The proof of \cthm\ref{thm:mdrq_correctness} is provided in \capp\ref{app:mdrq_correctness_proof}.

According to \cequ\eqref{eq:range_freq_est}, our mechanism estimates range query by by summing up all estimations of single values' frequencies in the range $R$.
Thus, it is natural to expect this approach of range query estimation to incur an expected squared error that is $\bigohinline{m^D}$ times larger than that of a single value, since the size of the range may be as large as $\bigohinline{m^D}$. Surprisingly, however, as we will show in the following, the range query's estimation error has the same upper bound as the single-value frequency's estimation error if $D_R=D$.
In contrast, existing methods for range query estimation \cite{kulkarni2019answering, wang2019answering} all have an amplification factor of $(\log m)^{\bigohinline{D}}$ on the expected squared error in comparison to that of a single-value frequency.
We have the following theorems on the accuracy.
%
%
\begin{theorem}[Single-value frequency]\label{thm:mdrq_variance}
    For any value $x\in [m]^D$, the expected squared error of estimation $\hat{c}_x$ is
    \[\epinline{\| \hat{c}_x-c_x \|^2} = \vrinline{\hat{c}_x} = \bigoh{ (\frac{e^\eps+1}{e^\eps-1})^{2D}2^{-D}
(1-(\frac{e^\eps-1}{e^\eps+1})^{2D})n}.\]
\end{theorem}
\begin{theorem}[Range query]\label{thm:mdrq_range_variance}
    For any range $R=[l_1,r_1]\times [l_2,r_2]\times\cdots\times[l_D,r_D]$ with $D_R$ non-trivial dimensions, the expected squared error of estimation $\hat{c}(R)$ is
    \[\epinline{\| \hat{c}(R)-c(R) \|^2} = \vrinline{\hat{c}(R)} = \bigoh{(\frac{e^\eps+1}{e^\eps-1})^{2D}2^{-D_R} (1-(\frac{e^\eps-1}{e^\eps+1})^{2D})n}.\]
\end{theorem}
When $\eps$ is small, $\vrinline{\hat{c}_x}\approx O( (\frac{2}{\eps^2})^D n )$, and $\vrinline{\hat{c}(R)}\approx O((\frac{2}{\eps})^{2D}2^{-D_R} n)$. The proof of \cthm\ref{thm:mdrq_variance} is provided in \capp\ref{app:mdrq_variance_proof}, and the proof of \cthm\ref{thm:mdrq_range_variance} is provided in \capp\ref{app:mdrq_range_variance_proof}.

Here, let's give some intuitive explanation on why the expected squared errors for both range query and single value have the same upper bound (if $D_R=D$).
%
%
%
This is from the nice property of the bias correction matrix $\mathbf{B}^{-1}_{m,D}$. More specifically, when calculating the estimation for the range query in \eqref{eq:range_freq_est},
the expected squared error of the estimation is affected by the non-zero terms in $\sum_{x\in R}\mathbf{e}_x \mathbf{B}^{-1}_{m,D}$, which is a summation of multiple rows of $\mathbf{B}^{-1}_{m,D}$ with each row corresponding to one point in $R$.
Fortunately, we can show that instead of exploding the number of non-zero terms in the summation by $\bigohinline{|R|}$, most of the terms are canceled out, leaving the number of remaining non-zero terms to be equal to that in a single-value frequency query. Therefore, the expected squared error is not amplified from single values to range queries. More details can be found in our proof in \capp\ref{app:mdrq_range_variance_proof}.

Notice that there is still a $\bigohinline{2^{D-D_R}}$ gap between the analyzed bounds for $\vrinline{\hat{c}(R)}$ and $\vrinline{\hat{c}_x}$ when $D_R<D$.
This gap can be easily removed by slightly changing our mechanism mentioned in Section \ref{sec:mdrq:alg}, which leads to $D_R=D$.
To enforce $D_R=D$ for any range query, we can extend the domain size $m$ of each dimension by adding one dummy value, \ie, changing the domain $[m]$ to $[m+1]$ for each dimension, even though no data point will lie in the extended extra space.  
Then, any range query in the original space $[m]^D$ will have $D_R=D$ since every dimension of the range is nontrivial in the extended space $[m+1]^D$, which leads to $\vrinline{\hat{c}(R)}= \vrinline{\hat{c}_x}$.


\stitle{Simulation results.}
%
We perform a simple simulation to evaluate the empirical error of our mechanism and verify our theoretical analysis (\cthms\ref{thm:mdrq_variance} and \ref{thm:mdrq_range_variance}).
We use synthetic data generated as follows. For any $D$-dimensional private value, each of its dimension follows Zipf distribution with parameter $1.1$.  We implement our mechanism in \csec\ref{sec:mdrq:alg} and measure the average squared error of frequency queries over all single values. We also randomly generate $100$ range queries and measure the average squared error of all these range queries. The mechanism (both encoding and estimation) is repeated three times.
The analyzed error bound is the one presented in Theorem \ref{thm:mdrq_variance} (with the constant set to be $1$ in the big oh, it is equal to the analytical upper bound as shown in the proof).
%
%
As we can observe from the Figure \ref{fig:mdrq}, both the empirical squared errors of single-value frequencies and range queries are below our analyzed error bound (Theorem \ref{thm:mdrq_variance}), proving the effectiveness of our mechanism.

\begin{figure}[htp]
    \begin{center}
        \begin{subfigure}[b]{0.4\textwidth}
                \includegraphics[width=\linewidth]{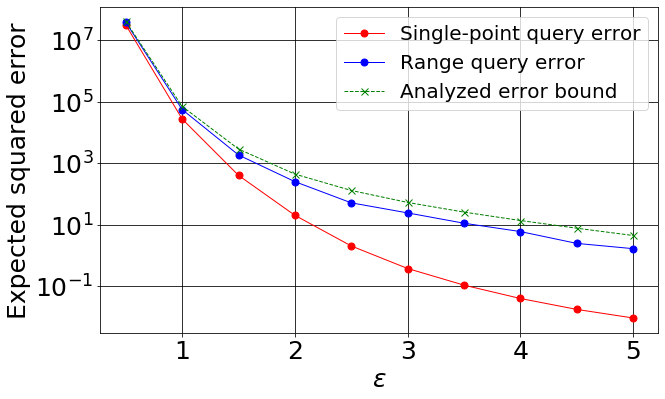}
                \caption{$D=5$, $n=1000$, and $m=10$}
                \label{fig:d=5}
        \end{subfigure}%
        \hspace{0.1\textwidth}
        \begin{subfigure}[b]{0.4\textwidth}
                \includegraphics[width=\linewidth]{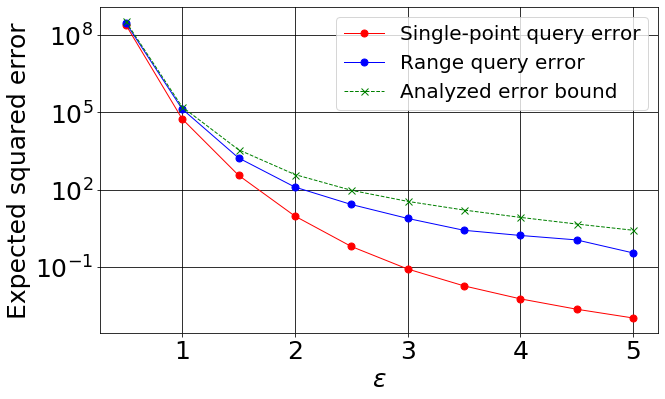}
                \caption{$D=6$, $n=1000$, and $m=10$}
                \label{fig:d=6}
        \end{subfigure}
        \caption{Squared error in estimating single-value frequencies and multi-dimensional range queries}\label{fig:mdrq}
    \end{center}
\end{figure}


\stitle{Computational and space complexity.}
Each data owner will compute and send $D$ vectors of length $m$ to the data collector, leading to computational/space complexity $\bigohinline{mD}$.

For the data collector, in order to obtain the frequency estimation vector $\rangeest=(\frac{e^\eps+1}{e^\eps-1})^{D}\mathbf{B}^{-1}_{m,D}\cdot\ob$, both $\mathbf{B}^{-1}_{m,D}$ and $\ob$ need to be computed. 
Since $\ob$ is a vector of size $m^D$ and computing each entry requires $\bigohinline{nD}$ time, the computational complexity for computing $\ob$ is $\bigohinline{Dm^Dn}$.
The computation of $\mathbf{B}^{-1}_{m,D}$ can be done via the recursion in \cequ\eqref{eq:recursion}. By the argument in the proof of Theorem \ref{thm:mdrq_variance}, each row of $\mathbf{B}^{-1}_{m,D}$ has only $2^D$ non-zero elements, and therefore computing frequency estimation for any value $x$ requires $\bigohinline{2^D}$ time. Then computing every value in the range $[m]^D$ requires $\bigohinline{2^Dm^D} = \bigohinline{(2m)^D}$ time.
There are three different implementations for the rest steps.
\begin{itemize}
    \item (store single-value frequencies): We can store frequency estimations for all values in $[m]^D$ with $\bigohinline{m^D}$ space, and answer a query $R$ by summing up $\hat{c}_x$'s for all $x\in R$ with $\bigohinline{|R|}$ time.
    \item (store all prefix sums): In addition, we can compute and store all the $\bigohinline{m^{D}}$ multi-dimensional prefix sums. A range query can be answered by combining $\bigohinline{2^D}$ answers of the prefix sums.
    \item (on-the-fly): Alternatively, we can skip the above preprocessing, and compute the answer for a range query on the fly. In this case, we need $\bigohinline{Dmn}$ to store all data owners' reports and $\bigohinline{\max(Dm^Dn, 2^D|R|)}$ time (computing $\ob$ and $|R|$ frequency estimations) to process a query.
\end{itemize}
%
%
%
%
We summarize the computational and space costs for the data collector in Table \ref{table:mdrq}.

\begin{table}[ht]
\centering
\caption{Summary of the computational and space costs of the data collector in different implementations}
\begin{tabular}[t]{lccc}
\toprule
&store single-value frequencies & store all prefix sums & on-the-fly\\
\midrule
{space cost}  & $\bigohinline{m^D}$ & $\bigohinline{m^D}$ & $\bigohinline{Dmn}$ \\
{preprocessing}  & $\bigohinline{\max(Dm^Dn,(2m)^D)}$ & $\bigohinline{\max(Dm^Dn,(2m)^D)}$ & $\bigohinline{1}$ \\
query cost & $\bigohinline{|R|}$ & $\bigohinline{2^D}$ & $\bigohinline{\max(Dm^Dn, 2^D|R|)}$ \\
\bottomrule
\end{tabular}
\label{table:mdrq}
\end{table}

\stitle{Handling continuous domains.}
In general, the input may be vectors from a real domain $[0, \Sigma]^D$. 
To apply the mechanism introduced in this section, a mapping from $[0, \Sigma]^D$ to $[m]^D$ is needed (\eg, partitioning each dimension $[0, \Sigma]$ evenly into sub-intervals and mapping each of them to a value in $[m]$).
At first glance, it is appealing to choose a larger $m$ for such a discretization process, since the {\em truncation error} (due to the rounding from $[0, \Sigma]$ to $[m]$) can be smaller as a more accurate range in $[m]$ can be used for the range query in $[0,\Sigma]$, while the error bounds in \cthms\ref{thm:mdrq_variance} and \ref{thm:mdrq_range_variance} are independent on $m$.
However, a larger $m$ means the ``real'' distance between $i$ and $i+1$ in $[m]$ is smaller in the original domain $[0, \Sigma]$, and thus a smaller $\eps$ is needed to guarantee the same level of privacy protection, resulting in a larger {\em estimation error} according to \cthms\ref{thm:mdrq_variance} and \ref{thm:mdrq_range_variance}.
%
%
Therefore, it is possible to choose an optimal value of $m$ to minimize the total error of the two types introduced above.
The optimal selection of $m$ may depend on the distribution of the input data, which is hard to be quantified, and is private, too. We leave it as an open question for future work.

\subsection{Extensions to More Complex Range Queries}
We introduce how our mechanism can be extended for more complex range queries: i) when each dimension has a different domain size; and ii) weighted range queries.

\subsubsection{When Dimension Sizes are Different}
\label{sec:mdrq:domain}
Our mechanism in Section \ref{sec:mdrq:alg} naturally extends to the case where each dimension has a different size, \ie, the private values are in domain $[m_1]\times\ldots\times[m_D]$.
We use similar notations as in \csec\ref{sec:mdrq:alg}, and we also assign an index to a value in the domain: $\ind(x)=1+\sum_{d=1}^D ((\prod_{j=1}^d m_{j-1}) (x[d]-1))$ where $m_0=1$.
If it is clear from the context, we will use $x$ to also denote its index $\ind(x)$.
%

Each data owner sends to the data collector $D$ vectors $\rmatrix=\vect{\mathbf{r}_1,...,\mathbf{r}_D}^\intercal$, where each $\mathbf{r}_d$ for dimension $d$ is a length-$m_d$ vector constructed the same way as in \csec\ref{sec:mdrq:alg}.

After collecting data owners' responses $\rmatrix_1,...,\rmatrix_n$, the data collector similarly calculates an observation vector of length $\prod_{j=1}^D m_j$ as 
\[
\obs_x = \sum_{i=1}^n \prod_{d=1}^D \rmatrix_i[d, x[d]], \quad \forall x \in [m_1]\times\ldots\times[m_D]
\]
where $\rmatrix_i[d, x[d]]$ denotes $\mathbf{r}_d[x[d]]$ for $\rmatrix_i=\vect{\mathbf{r}_1,...,\mathbf{r}_D}^\intercal$.  
Then the data collector can estimate the frequencies of all single values in $[m_1]\times\ldots\times[m_D]$ as a vector
\begin{equation*}
    \rangeest=
    (\frac{e^\eps+1}{e^\eps-1})^{D}\mathbf{C}^{-1}_{D}\cdot \ob
\end{equation*}
where $\mathbf{C}^{-1}_{D}$ is a $(\prod_{j=1}^D m_j)\times (\prod_{j=1}^D m_j)$ matrix that can be partitioned into $m_D\times m_D$ submatrices, defined by the following recursive relation for $2\leq k\leq D$:
\begin{equation*}
\mathbf{C}_{k}^{-1}=
\frac{1}{2}
\begin{bmatrix}
\mathbf{C}_{k-1}^{-1} & 0 & \cdots  & 0 & \mathbf{C}_{k-1}^{-1}\\
-\mathbf{C}_{k-1}^{-1} & \mathbf{C}_{k-1}^{-1} & \ddots & \vdots & 0\\
0 & -\mathbf{C}_{k-1}^{-1} & \ddots & 0 & \vdots\\
\vdots & \ddots & \ddots & \mathbf{C}_{k-1}^{-1} & 0\\
0 & \dots & 0 & -\mathbf{C}_{k-1}^{-1} & \mathbf{C}_{k-1}^{-1}\\
\end{bmatrix},
\quad
\mathbf{C}_{1}^{-1}=
\frac{1}{2}
\underbrace{
\begin{bmatrix}
1 & 0 & \cdots  & 0 & 1\\
-1 & 1 & \ddots & \vdots & 0\\
0 & -1 & \ddots & 0 & \vdots\\
\vdots & \ddots & \ddots & 1 & 0\\
0 & \dots & 0 & -1 & 1\\
\end{bmatrix}
}_{m_1}.
\end{equation*}
$\mathbf{C}_{k-1}^{-1}$ is a $(\prod_{j=1}^{k-1} m_j)\times (\prod_{j=1}^{k-1} m_j)$ matrix. For the base case, $\mathbf{C}_{1}^{-1}$ is the above $m_1\times m_1$ matrix.

Let $R=[l_1,r_1]\times [l_2,r_2]\times\cdots\times[l_D,r_D]$ be the range for a range query.
We can similarly obtain the answer for the range query by directly summing up the frequencies of all $x\in R$, that is,
\begin{equation}
    \hat{c}(R)=\sum_{x\in R}\hat{c}_x=
    (\frac{e^\eps+1}{e^\eps-1})^{D}\sum_{x\in R}\mathbf{e}_x\mathbf{C}^{-1}_{D}\cdot \ob
\end{equation}

The correctness proof and accuracy analysis are similar to the case with identical domain sizes. The error bounds of the estimations are identical to the bounds in \cthms\ref{thm:mdrq_variance} and \ref{thm:mdrq_range_variance}.

\subsubsection{Weighted Range Queries}
\label{sec:mdrq:weighted}
When each data owner $i$ also holds a (private) weight $w_i\in W$, a {\em weighted range query} asks $c_{\bf w}(R)=\sum_{i=1}^n w_i \indicator{x_i\in R}$.
Without loss of generality, we consider weights from the domain $W = [0, \Delta]$, where $\Delta$ is public knowledge, as we can shift other intervals with finite lengths to this domain by adding constants.
%
%
Two cases are considered below, for non-private and private weights, respectively.

\stitle{If the weights are non-private information}, results for unweighted multi-dimension range query can be easily extended to weighted multi-dimensional range query.
All data owners are partitioned into groups by their weights, and each group is formed by data owners that has identical weights.
For each group $g_w$ with weight $w$, we construct an estimator for unweighted multi-dimensional range queries $\hat{c}_{g_w}(\cdot)$.
To answer a weighted range query, we sum up the weighted answers from all groups,
\[
    \hat{c}_{\bf w}(R)=\sum_{w\in W}w\cdot \hat{c}_{g_w}(R).
\]

\stitle{If the weights are private information},
we can consider the weight as an extra private dimension for each data owner, and then use unweighted $(D+1)$-dimensional range query oracle to answer weighted $D$-dimensional range queries. More specifically, given weights ${\bf w}=[w_1, \ldots, w_n]$ for $n$ data owners, each owner $i$ first rounds her private weight $w_i$ to $\{1,2\}$: to $2$ with probability $\frac{w_i}{\Delta}$ and to $1$ with probability $\frac{\Delta-w_i}{\Delta}$. The rounded private weight is considered to be the $(D+1)$-th private dimension with domain $\{1,2\}$. The same encoding algorithm in \csec\ref{sec:mdrq:alg} is used. For the data collector, we construct an estimator $\hat{c}(\cdot)$ for unweighted $(D+1)$-dim range queries, and estimate the answer to the weighted range query with $R=[l_1,r_1]\times\cdots\times[l_D,r_D]$ as
\[
    \hat{c}_{\bf w}(R)=\Delta\cdot \hat{c}(R\times[2, 2]).
\]
The estimation is unbiased since $\epinline{\hat{c}_{\bf w}(R)}=\epinline{\Delta\cdot \hat{c}(R\times[2,2])}=\Delta\cdot \epinline{\hat{c}(R\times[2,2])}=\Delta\cdot \sum_{i\in[n]}(\frac{x_i}{\Delta}\cdot \indicator{x_i\in R})=\sum_{i\in[n]}x_i\indicator{x_i\in R}$. Error in these two cases can be bounded as follows.
\begin{theorem}\label{thm:mdrq:weighted}
    For weighted multi-dimensional range query with any range $R=[l_1,r_1]\times [l_2,r_2]\times\cdots\times[l_D,r_D]$ of dimension $D_R$, the expected squared error (variance) of estimation $\hat{c}_{\bf w}(R)$ is
    \begin{enumerate}
        \item (non-private weights) $\epinline{\| \hat{c}_{\bf w}(R)-c_{\bf w}(R) \|^2} = \bigohinline{(\frac{e^\eps+1}{e^\eps-1})^{2D}2^{-D_R} (1-(\frac{e^\eps-1}{e^\eps+1})^{2D})\Delta^2 n}$, or
        \item (private weights) $\epinline{\| \hat{c}_{\bf w}(R)-c_{\bf w}(R) \|^2} = \bigohinline{(\frac{e^\eps+1}{e^\eps-1})^{2(D+1)}2^{-(D_R+1)} (1-(\frac{e^\eps-1}{e^\eps+1})^{2(D+1)})\Delta^2 n}$.
    \end{enumerate}
\end{theorem}
The proof of \cthm\ref{thm:mdrq:weighted} is in \capp\ref{app:thm:mdrq:weighted}. The same technique of extending each dimension by a dummy element mentioned in \csec\ref{sec:mdrq:err} can be applied to replace $D_R$ with $D$ in the bounds.

\subsection{Application: Quantile Queries}
\label{sec:mdrq:quantile}
Consider quantile queries in a one-dim domain $\domain = [m]$. 
The {\em percentile} of a value $x$ in $\data = \{x_i\}_{i \in [n]}$ is $\sigma(x) = \frac{1}{n} \sum_{i=1}^n \indicator{x_i \leq x}$, which calculates the fraction of values that are no larger than $x$ in $X$. 
The interval $I(x)=(\sigma(x-1),\sigma(x)] \subseteq [0, 1]$ is said to be the {\em percentile interval} of $x$.
The {\em $p$-quantile} of $X$ is defined to be the value $x^*$, such that $\sigma(x^*-1) < p \leq \sigma(x^*)$, \ie, $p$ is in $x^*$'s percentile interval. Note that median estimation is a special case of the quantile query when $p=0.5$.

Let $\hat{x}^*$ be an estimated $p$-quantile. The goal is to make sure that $p$ is close to $I({\hat{x}^*})$ (or, the percentile of $\hat{x}^*$ is close to $p$). We define the error of the estimation $\hat{x}^*$ to be 
\[
\Err[\hat{x}^*]=\inf_{\hat{p} \in I({\hat{x}^*})} |\hat{p} - p|.  
\]
We want to bound error $\Err[\hat{x}^*]$ with high probability.
Note that our error definition is essentially the $\epsilon$-approximate $p$-quantile in literature, \eg, \cite{manku1998approximate} (with $\inf$ considered here as $\data$ is a multiset).
%
%

\stitle{Answering quantile queries under $E_{\lone}$-LDP.}
We consider the metric $E_{\lone}$ defined for range queries.
Our mechanism of quantile query estimation basically follows the approach proposed in Section 4.7 of \cite{kulkarni2019answering}, which uses one-dimensional range query mechanism as a primitive and perform binary search to estimate the $p$-quantile. Our main contribution here is to provide formal analysis on the utility of the mechanism, and compare the mechanisms under $\eps$-LDP and $E_{\lone}$-LDP.

For data owners, private values are encoded using the algorithm (its $1$-dim case) in \csec\ref{sec:mdrq:alg} to guarantee $E_{\lone}$.
For the data collector, let $\hat{c}([l,r])$ be the frequency of range $[l,r]$ estimated using the mechanism introduced in \csec\ref{sec:mdrq:alg} for answering one-dim range queries. We can then estimate the percentile of value $x$ as $\hat{\sigma}(x)=\hat{c}([1, x])/n$.
Our mechanism answers a $p$-quantile query as follows.
\begin{enumerate}
    \item Construct an oracle (\csec\ref{sec:mdrq:alg}) for answering $1$-dim range queries on data owners' reports.
    \item Perform binary search on the input data domain $[m]$ until a value $\hat{x}^*$ s.t. $\hat{\sigma}(\hat{x}^*-1)< p \leq \hat{\sigma}(\hat{x}^*)$ is found, with $\hat{\sigma}$ defined above. More specifically, initially let $L=1, R=m$ and $M=\lceil(L+R)/2\rceil$. If $\hat{\sigma}(M)<p$, then let $L=M$, otherwise let $R=M$. Let $M=\lceil(L+R)/2\rceil$ and repeat the above procedure until find the value $\hat{x}^*$ (if $R-L \leq 10$, we perform a linear-scan search for $\hat{x}^*$).
    \item Output $\hat{x}^*$ as the estimation for the quantile query.
\end{enumerate}

\stitle{Accuracy analysis.}
We show that, with high probability, the estimation error $\Err[\hat{x}^*]$ is bounded.
\begin{lemma}\label{lem:qq}
    With probability at least $1-\delta$, the error of estimated percentile is bounded
    \[
        |\hat{\sigma}(x)-\sigma(x)|\leq \frac{e^\eps+1}{e^\eps-1}\cdot \sqrt{\frac{2}{n}\log \frac{1}{\delta}}, \quad \hbox{for any $x \in [m]$}.
    \]
\end{lemma}
\begin{theorem}\label{thm:qq}
    With probability at least $1-\delta$, our quantile query mechanism guarantees that
    \[
        \Err[\hat{x}^*]\leq \frac{2(e^\eps+1)}{e^\eps-1} \cdot \sqrt{\frac{2}{n}\log \frac{2\log m}{\delta}}, \quad \hbox{for an estimated $p$-quantile $\hat{x}^*$}.
    \]
\end{theorem}
The proofs of \clem\ref{lem:qq} and \cthm\ref{thm:qq} are in \capp\ref{app:lem:qq} and \capp\ref{app:thm:qq}, respectively.
    
\stitle{Comparison with the mechanism under $\eps$-LDP.} From \cthm\ref{thm:qq}, the estimation error of our mechanism is bounded by $\bigohinline{\frac{1}{\eps\sqrt{n}}\sqrt{\log\log m}}$ with high probability. 
The state-of-the-art $\eps$-LDP mechanisms \cite{kulkarni2019answering, wang2019answering} for one-dim range queries can be plugged in step 1 (as suggested in \cite{kulkarni2019answering}). Since these primitives have error bounded by $\bigohinline{\frac{1}{\eps\sqrt{n}}\log m}$ with high probability for estimating percentiles, following the same argument, the error of the $p$-quantile estimation is $\bigohinline{\frac{1}{\eps\sqrt{n}}\log m\sqrt{\log\log m}}$, which is $\bigohinline{\log m}$ times larger compared to our mechanism under $E_{\lone}$-LDP.


\begin{table}[ht]
\begin{center}
\caption{Summary of error in $E_{\lone}$-LDP algorithms and $\eps$-LDP algorithms for different tasks ($\eps$ is small)}
\begin{tabular}{l l l l}
\toprule
& 1-dim range query & Multi-dim range query & Quantile query
\\
& (expected squared error) & (expected squared error) & (error w.h.p.)
\\ \midrule
\cite{kulkarni2019answering, wang2019answering} ($\eps$-LDP) & $\bigohinline{\frac{n \log m}{\eps^2}}$ & $\bigohinline{\frac{n (\log m)^{D+D_R} }{\eps^2}}$ & $\bigohinline{\frac{\log m}{\eps\sqrt{n}}\sqrt{\log\log m}}$
\\ 
Extending \cite{haney2015design} ($E_{\lone}$-LDP) & $\bigohinline{\frac{n}{\eps^2}}$ & $\bigohinline{\frac{nD (\log m)^{2(D-1)}}{\eps^2}}$ & ?
\\ 
This work ($E_{\lone}$-LDP) & $\bigohinline{\frac{n}{\eps^2}}$ & $\bigohinline{n(\frac{2}{\eps})^{2D}2^{-D_R}}$ & $\bigohinline{\frac{1}{\eps\sqrt{n}}\sqrt{\log\log m}}$
\\ \bottomrule
\end{tabular}
\label{table:summary_accuracy}
\end{center}
\end{table}

\section{Conclusion}
This paper investigates local differential privacy on metric spaces (or $E$-LDP), which is a relaxation of $\eps$-LDP to customize the levels of indistinguishability among different pairs of values using a metric function $E$.
In this work, we design a generic $E$-LDP mechanism (generalizing matrix mechanisms in CDP) to trade-off privacy for utility of linear counting queries. For multi-dimensional range queries, we introduce a novel $E$-LDP algorithm under $\lone$-metric with an error which is independent on the size $m$ of each dimension. This technique can also help reduce the error of $\eps$-LDP algorithms for quantile queries by a factor of $\log m$ under $E$-LDP. 
Our techniques apply to $\lpnorm$-LDP as well.
As future work, we would apply techniques in this paper as primitives for other analytical tasks; we would also expect the transform matrices developed in this paper to be used to improve algorithms in the centralized setting under similar relaxations (\eg, in Blowfish privacy).

\bibliographystyle{plain}
\bibliography{references}


\appendix

\section{Missing Proofs}
\label{app:missing_proofs}

\subsection{Proof for \cprop\ref{prop:metric_space}}
\label{prop:metric_space:proof}
Clearly $E(x,x)=0$ and $E(x,x')=E(x',x)$. Now we show the triangle inequality. We want to show that, if for some $x,y,z\in \mathcal{X}$, $E(x,z) > E(x,y)+E(y,z)$, then $E$ is not tight. By definition
\[
\frac{\Pr[\mathcal{A}(x)\in S]}{\Pr[\mathcal{A}(z)\in S]}=\frac{\Pr[\mathcal{A}(x)\in S]}{\Pr[\mathcal{A}(y)\in S]}\cdot \frac{\Pr[\mathcal{A}(y)\in S]}{\Pr[\mathcal{A}(z)\in S]}\leq 
e^{E(x,y)+E(y,z)}<e^{E(x,z)}
\]
which shows $E$ is not tight for $\ldpalgo$.
Therefore for any tight policy $E$, we have
$
E(x,z)\leq E(x,y)+E(y,z)
$.
\eop

\subsection{Proof for \cprop\ref{prop:mmerror}}
\label{prop:mmerror:proof}
    $\epinline{\estimate(\ldpdata)} = \epinline{\mathbf{B} \cdot \sum_{i=1}^n \mathbf{r}_i}=\epinline{\mathbf{B} \cdot \sum_{i=1}^n (\mathbf{A} \cdot \mathbf{h}_x + \lap({\bf s}))}=\mathbf{W}\sum_{i=1}^n \mathbf{h}_x + \epinline{\mathbf{B} \cdot \sum_{i=1}^n \lap({\bf s})}={\bf W} \cdot {\bf c}$, therefore the estimation is unbiased.
    
    To show the variance of the estimation, we have
    $
    \epinline{\|\estimate(\ldpdata) - {\bf W}\cdot{\bf c}\|^2}=
    \epinline{\|\mathbf{B} \cdot \sum_{i=1}^n \lap({\bf s})\|^2}=
    \trace[\mathbf{B} \cdot \vrinline{\sum_{i=1}^n \lap({\bf s})} \cdot \mathbf{B}^\intercal]=
    2n \cdot \trace[\mathbf{B}^\intercal\mathbf{B} \cdot \diag(s_1^2, \ldots, s_p^2)]
    $. The last equality is due to the property of Laplace distribution,  $\vrinline{\lap({\bf s})}=2\diag(s_1^2, \ldots, s_p^2)$.
\eop

\subsection{Proof for \cthm\ref{thm:mdrq_correctness}}
\label{app:mdrq_correctness_proof}
%
%
We will first show that the estimate for the frequency of all values among the $n$ users is unbiased
i.e., $\epinline{\hat{\counts}}=\counts$,  where $\hat{\mathbf{c}}=(\frac{e^{\eps}+1}{e^{\eps}-1})^D\mathbf{B}^{-1}_{m,D}\cdot \ob$, where $\ob$ is a vector of observations defined by Eqn.~\ref{eq:observation} and $\mathbf{B}^{-1}_{m,D}$ is a bias correction matrix defined by Eqn.~\ref{eq:recursion}. Since the answer for a range query is the sum of a set of unbiased estimates, it is also unbiased. 

Given the input $x_i$ held by user $i$, the encoding of the $d$-th dimension of $x_i$ before perturbation is a $m$-length vector, $\mathbf{b}_d\in \{-1,1\}^m$, where the first $(v-1)$ bits are -1 and the rest are 1 if $x_i[d]=v$. The perturbation algorithm flips these bits independently with probability $p=\frac{1}{1+e^{\eps}}$ and results a new vector $\mathbf{r}_d$, i.e., $\Pr[\mathbf{r}_d[k] \leftarrow \mathbf{b}_d[k]]=1-p$ and $\Pr[\mathbf{r}_d[k] \leftarrow (-1\cdot\mathbf{b}_d[k])]=p$, for every $k\in[m]$.
Then the report matrix $\mathbf{R}_i =\vect{ \mathbf{r}_1,\mathbf{r}_2,...,\mathbf{r}_D}^\intercal$ is sent by user $i$ to the data collector. The observation for $x\in [m]^D$ is $\obs_{x}= \sum_{i=1}^n \prod_{d=1}^D \mathbf{R}_i[d, x[d]]$.

Let's analyze the expected value for the contribution of each user $i$ in $\obs_x$, i.e., $\prod_{d=1}^D \mathbf{R}_i[d, x[d]]$. Before perturbation, the corresponding product is $\prod_{d=1}^D \mathbf{b}_d[x[d]]$. Note that if $x_i[d]>x[d]$, then $\mathbf{b}_d[x[d]]=-1$. Hence, this product depends only on the number of the coordinates of $x_i$ which have larger value than $x$, denoted by $S_{>}(x_i,x)=\sum_{d=1}^D \mathbbm{1}_{x_i[d]>x[d]}$. Hence, we have $\prod_{d=1}^D \mathbf{b}_d[x[d]]= (-1)^{S_{>}(x_i,x)}$.
Now consider how the above product changes after perturbation. 
Since each position of the vector $\mathbf{b}$ is $\{-1, 1\}$, flipping even number of the positions in $x[1],...,x[D]$ will not change the product. Hence, with probability $P_{even}=\sum\limits_{i\in[0,D], i\text{ even}}\binom{D}{i}p^i(1-p)^{D-i}$, we have 
\[\prod_{d=1}^D \mathbf{R}_i[d, x[d]]=\prod_{d=1}^D \mathbf{b}_d[x[d]] = (-1)^{S_{>}(x_i,x)};\]
and with probability $P_{odd}=\sum\limits_{i\in[0,D], i\text{ odd}}\binom{D}{i}p^i(1-p)^{D-i}$, we have 
\[\prod_{d=1}^D \mathbf{R}_i[d, x[d]]=-\prod_{d=1}^D \mathbf{b}_d[x[d]] = (-1)\cdot(-1)^{S_{>}(x_i,x)}.\]
Therefore, 
$\epinline{\prod_{d=1}^D \mathbf{R}_i[d, x[d]]} = 
(P_{even}-P_{odd}) (-1)^{S_{>}(x_i,x)} = (1-2p)^D (-1)^{S_{>}(x_i,x)}
$. The last equality is due to the fact that the expansion of $((1-p) -p)^{D}$ equals to $P_{even}-P_{odd}$.

When aggregating all users' contributions, we have
\[
\epinline{\ob_x}=\sum_{i=1}^n \epinline{\prod_{d=1}^D \mathbf{R}_i[d, x[d]]}= (1-2p)^D \sum_{i=1}^n (-1)^{S_{>}(x_i,x)}
=(1-2p)^D \sum_{x'=1}^{m^D} (-1)^{S_{>}(x',x)} c_{x'}
\]
The last equality is simply by changing iterating data owners' value to iterating values in the $D$-dim domain.
Thus, by the above derivation, we obtained 
a linear transformation from the frequency vector $\mathbf{c}=[c_1,...,c_{m^D}]$ to the observed vector $\ob=[\obs_1,...,\obs_{m^D}]$, characterized by a $m^D\times m^D$ matrix $\mathbf{A}_{m,D}=(1-2p)^D\mathbf{B}_{m,D}$, where $\mathbf{B}_{m,D}$ is a matrix whose elements are either $-1$ or $1$.

Therefore, the relation can be written as
$\epinline{\ob} =(1-2p)^D\mathbf{B}_{m,D}\cdot \mathbf{c}$ where $p=\frac{1}{1+e^{\eps}}$.
Together with \clem\ref{lem:inverse}, the estimate $\hat{\mathbf{c}}=(\frac{e^{\eps}+1}{e^{\eps}-1})^D\mathbf{B}^{-1}_{m,D}\cdot \ob$ for $\mathbf{c}$ is unbiased.
\eop

\begin{lemma}\label{lem:inverse}
    The matrix $\mathbf{B}^{-1}_{m,D}$ defined by \cequ\eqref{eq:recursion} in Section \ref{sec:mdrq:alg}
    is the matrix inverse of $\mathbf{B}_{m,D}$.
\end{lemma}
\begin{proof}
It can be observed that $\mathbf{B}_{m,D}$ can be partitioned into $m\times m$ submatrices $\mathbf{B}_{m,D-1}$, satisfying the following recursive relation for $2\leq d \leq D$,
\begin{equation*}
\mathbf{B}_{m,d}=
\begin{bmatrix}
\mathbf{B}_{m,d-1} & -\mathbf{B}_{m,d-1} & \cdots  & -\mathbf{B}_{m,d-1}\\
\vdots & \ddots & \ddots & \vdots\\
\vdots &  & \ddots & -\mathbf{B}_{m,d-1}\\
\mathbf{B}_{m,d-1} & \cdots & \cdots & \mathbf{B}_{m,d-1}\\
\end{bmatrix}
\end{equation*}
That is, after partition, the submatrices in the bottom-left triangle are all $\mathbf{B}_{m,d-1}$ and rest of the submatrices are all $-\mathbf{B}_{m,d-1}$. For the base case when $D=1$,
\begin{equation*}
\mathbf{B}_{m,1}=
\begin{bmatrix}
1 & -1 & \cdots  & -1\\
\vdots & \ddots & \ddots & \vdots\\
\vdots &  & \ddots & -1\\
1 & \cdots & \cdots & 1\\
\end{bmatrix}
\end{equation*}

Recall that $\mathbf{B}^{-1}_{m,D}$ is defined by the following recursive relation in Eqn. \ref{eq:recursion}: for $2\leq d \leq D$,
\begin{equation*}
\mathbf{B}_{m,d}^{-1}=
\frac{1}{2}
\begin{bmatrix}
\mathbf{B}_{m,d-1}^{-1} & 0 & \cdots  & 0 & \mathbf{B}_{m,d-1}^{-1}\\
-\mathbf{B}_{m,d-1}^{-1} & \mathbf{B}_{m,d-1}^{-1} & \ddots & \vdots & 0\\
0 & -\mathbf{B}_{m,d-1}^{-1} & \ddots & 0 & \vdots\\
\vdots & \ddots & \ddots & \mathbf{B}_{m,d-1}^{-1} & 0\\
0 & \dots & 0 & -\mathbf{B}_{m,d-1}^{-1} & \mathbf{B}_{m,d-1}^{-1}\\
\end{bmatrix}
\end{equation*}
where $\mathbf{B}_{m,d-1}^{-1}$ is a $m^{d-1}\times m^{d-1}$ submatrix. For the base case, $\mathbf{B}_{m,1}^{-1}$ is the following $m\times m$ matrix
\[
\mathbf{B}_{m,1}^{-1}=
\frac{1}{2}
\begin{bmatrix}
1 & 0 & \cdots  & 0 & 1\\
-1 & 1 & \ddots & \vdots & 0\\
0 & -1 & \ddots & 0 & \vdots\\
\vdots & \ddots & \ddots & 1 & 0\\
0 & \dots & 0 & -1 & 1\\
\end{bmatrix}
\]

One can easily verify that $\mathbf{B}_{m,d}\mathbf{B}^{-1}_{m,d}=\diag(\mathbf{I}_{m^{d-1}},...,\mathbf{I}_{m^{d-1}})=\mathbf{I}_{m^{d}}$ for any $1\leq d \leq D$, and thus 
$\mathbf{B}^{-1}_{m,D}$ defined by Eqn. \ref{eq:recursion} is the matrix inverse of $\mathbf{B}_{m,D}$.
\end{proof}

\subsection{Proof for \cthm\ref{thm:mdrq_variance}}
\label{app:mdrq_variance_proof}
First we claim that each row of the matrix $\mathbf{B}^{-1}_{m,D}$ has $2^D$ elements, and each element is either $-(\frac{1}{2})^D$ or $(\frac{1}{2})^D$.
We can prove the claim by induction. 
    
First consider the base case when $D=1$,
\begin{equation*}
\mathbf{B}_{m,1}^{-1}=
\frac{1}{2}
\begin{bmatrix}
1 & 0 & \cdots  & 0 & 1\\
-1 & 1 & \ddots & \vdots & 0\\
0 & -1 & \ddots & 0 & \vdots\\
\vdots & \ddots & \ddots & 1 & 0\\
0 & \dots & 0 & -1 & 1\\
\end{bmatrix}
\end{equation*}
Thus the claim is true for $D=1$.

Suppose the claim is true for $D=k$, namely 
each row of the matrix $\mathbf{B}^{-1}_{m,k}$ has $2^k$ elements, and each element is $-(\frac{1}{2})^k$ or $(\frac{1}{2})^k$. 
Now we consider the case when $D=k+1$.
According to Equation~\eqref{eq:recursion}, we have
\begin{equation*}
\mathbf{B}_{m,k+1}^{-1}=
\frac{1}{2}
\begin{bmatrix}
\mathbf{B}_{m,k}^{-1} & 0 & \cdots  & 0 & \mathbf{B}_{m,k}^{-1}\\
-\mathbf{B}_{m,k}^{-1} & \mathbf{B}_{m,k}^{-1} & \ddots & \vdots & 0\\
0 & -\mathbf{B}_{m,k}^{-1} & \ddots & 0 & \vdots\\
\vdots & \ddots & \ddots & \mathbf{B}_{m,k}^{-1} & 0\\
0 & \dots & 0 & -\mathbf{B}_{m,k}^{-1} & \mathbf{B}_{m,k}^{-1}\\
\end{bmatrix}
\end{equation*}
Notice that each row of the matrix $\mathbf{B}_{m,k+1}^{-1}$ consists values from the rows of the two submatrices $\mathbf{B}_{m,k}^{-1}$ and $-\mathbf{B}_{m,k}^{-1}$, 
which by induction assumption has $2^k$ elements from $\{-(\frac{1}{2})^k, (\frac{1}{2})^k\}$.
Thus, each row of $\mathbf{B}_{m,k+1}^{-1}$ has $2^{k+1}$ elements, and each element is $-(\frac{1}{2})^{k+1}$ or $(\frac{1}{2})^{k+1}$. Therefore, the claim is true for any $D$.

Now we can prove the theorem. 
According to Equation~\eqref{eq:freq_est},
$\hat{c}_x=(\frac{e^\eps+1}{e^\eps-1})^{D}(\mathbf{e}_x\cdot \mathbf{B}^{-1}_{m,D})\cdot\ob$, 
where $\mathbf{e}_x$ denotes a binary vector with $x$-th position being $1$, thus $\mathbf{e}_x\cdot \mathbf{B}^{-1}_{m,D}$ computes the $x$-th row of matrix $\mathbf{B}^{-1}_{m,D}$. Therefore the variance can be computed as
\[
\vrinline{\hat{c}_x}=(\frac{e^\eps+1}{e^\eps-1})^{2D}
\vrinline{(\mathbf{e}_x \mathbf{B}^{-1}_{m,D})\cdot \ob}
\]

With the claim that each row of the matrix $\mathbf{B}^{-1}_{m,D}$ has $2^D$ elements, and each element is either $-(\frac{1}{2})^D$ or $(\frac{1}{2})^D$, we obtain
\[
\vrinline{\hat{c}_x}\leq (\frac{e^\eps+1}{e^\eps-1})^{2D}\cdot 
2^D (\frac{1}{2})^{2D}\cdot \max_x(\vrinline{\obs_x})
=(\frac{e^\eps+1}{e^\eps-1})^{2D}\cdot 
2^{-D} \cdot \max_x(\vrinline{\obs_x})
\]

Recall that by Eqn.~\ref{eq:observation} and the proof of Theorem \ref{thm:mdrq_correctness},
\[
\obs_x = \sum_{i=1}^n \prod_{d=1}^D \rmatrix_i[d,x[d]]
\]
where 
$\prod_{d=1}^D \rmatrix_i[d,x[d]]= (-1)^{S_{>}(x_i,x)}$ with probability
$P_{even}(D)$, and
$\prod_{d=1}^D \rmatrix_i[d,x[d]]= (-1)\cdot(-1)^{S_{>}(x_i,x)}$ with probability
$P_{odd}(D)$.
Thus, 
$\epinline{\obs_x}^2=n(P_{odd}(D)-P_{even}(D))^2=n(1-2p)^{2D}=n(\frac{e^\eps-1}{e^\eps+1})^{2D}$, and 
$\epinline{\obs_x^2}=n(P_{odd}(D)+P_{even}(D))=n$. We can calculate the variance as
\[
\vrinline{\obs_x}=\epinline{\obs_x^2}-\epinline{\obs_x}^2=(1-(\frac{e^\eps-1}{e^\eps+1})^{2D})n
\]
which leads to
\[
\vrinline{\hat{c}_x}= O\Big( (\frac{e^\eps+1}{e^\eps-1})^{2D}2^{-D}
(1-(\frac{e^\eps-1}{e^\eps+1})^{2D})n \Big).
\]
\eop

\subsection{Proof for \cthm\ref{thm:mdrq_range_variance}}
\label{app:mdrq_range_variance_proof}
According to Equation~\eqref{eq:range_freq_est}, we have 
$\hat{c}(R)=\sum_{x\in R}\hat{c}_x=
(\frac{e^\eps+1}{e^\eps-1})^{D}\sum_{x\in R}((\mathbf{e}_x \mathbf{B}^{-1}_{m,D})\cdot\ob)
=(\frac{e^\eps+1}{e^\eps-1})^{D}(\sum_{x\in R}(\mathbf{e}_x \mathbf{B}^{-1}_{m,D}))\cdot\ob$, where $\mathbf{e}_x \mathbf{B}^{-1}_{m,D}$ is the $x$-th row of matrix $\mathbf{B}^{-1}_{m,D}$.
We first examine the non-zero entries of
$\sum_{x\in R}(\mathbf{e}_x \mathbf{B}^{-1}_{m,D})$, which we denote as $\mathcal{N}(\sum_{x\in R}(\mathbf{e}_x \mathbf{B}^{-1}_{m,D}))$.

Recall that by Equation~\eqref{eq:recursion},
\begin{equation*}
\mathbf{B}_{m,D}^{-1}=
\frac{1}{2}
\begin{bmatrix}
\mathbf{B}_{m,D-1}^{-1} & 0 & \cdots  & 0 & \mathbf{B}_{m,D-1}^{-1}\\
-\mathbf{B}_{m,D-1}^{-1} & \mathbf{B}_{m,D-1}^{-1} & \ddots & \vdots & 0\\
0 & -\mathbf{B}_{m,D-1}^{-1} & \ddots & 0 & \vdots\\
\vdots & \ddots & \ddots & \mathbf{B}_{m,D-1}^{-1} & 0\\
0 & \dots & 0 & -\mathbf{B}_{m,D-1}^{-1} & \mathbf{B}_{m,D-1}^{-1}\\
\end{bmatrix}
\end{equation*}

Let $[\mathbf{B}_{m,D}^{-1}]_k$ denote the $k$-th row of $\mathbf{B}_{m,D}^{-1}$ after partitioning into submatrices. For instance,
$[\mathbf{B}_{m,D}^{-1}]_1=[\mathbf{B}_{m,D-1}^{-1}, 0, ..., 0, \mathbf{B}_{m,D-1}^{-1}]$,
$[\mathbf{B}_{m,D}^{-1}]_2=[-\mathbf{B}_{m,D-1}^{-1}, \mathbf{B}_{m,D-1}^{-1},  0, ..., 0]$.
Then
\[
\sum_{x\in R}(\mathbf{e}_x \mathbf{B}^{-1}_{m,D})=
\sum_{x'\in R/[l_D,r_D]} \mathbf{e}_{x'} \sum_{k=l_D}^{r_D}[\mathbf{B}^{-1}_{m,D}]_k
\]
where $R/[l_D,r_D]=[l_1,r_1]\times\cdots\times[l_{D-1},r_{D-1}]$.
Notice that $\sum_{k=l_D}^{r_D}[\mathbf{B}^{-1}_{m,D}]_k$ has only $1$ or $2$ submatrices remaining, since many submatrices will be canceled after summation.
More specifically, 
\[
\sum_{k=1}^{t}[\mathbf{B}^{-1}_{m,D}]_k=\frac{1}{2}[\underbrace{0,...,0}_{t-1}, \mathbf{B}^{-1}_{m,D-1}, 0, ...,0, \mathbf{B}^{-1}_{m,D-1}], \quad t=1, ..., m-1,
\]
\[
\sum_{k=1}^{m}[\mathbf{B}^{-1}_{m,D}]_k=\frac{1}{2}[0,...,0, 2\mathbf{B}^{-1}_{m,D-1}], \quad \hbox{and}
\]
\[
\sum_{k=s}^{t}[\mathbf{B}^{-1}_{m,D}]_k=\frac{1}{2}[\underbrace{0,...,0}_{s-2}, -\mathbf{B}^{-1}_{m,D-1}, 0, ...,0, \mathbf{B}^{-1}_{m,D-1},  \underbrace{0,...,0}_{m-t}], \quad 2\leq s\leq t\leq m.
\]
Thus $\sum_{k=l_D}^{r_D}[\mathbf{B}^{-1}_{m,D}]_k$ only consists submatrices $\mathbf{B}^{-1}_{m,D-1}$ or $-\mathbf{B}^{-1}_{m,D-1}$, and the non-zero entries of 
\[
\sum_{x\in R}(\mathbf{e}_x \mathbf{B}^{-1}_{m,D})= \sum_{x'\in R/[l_D,r_D]} \mathbf{e}_{x'}\sum_{k=l_D}^{r_D}[\mathbf{B}^{-1}_{m,D}]_k
\]
can be written as (denoted by $\mathcal{N}(\cdot)$):
\begin{align*}
&
\mathcal{N}(\sum_{x\in R}(\mathbf{e}_x \mathbf{B}^{-1}_{m,D}) =
\\
& \!\!\!\!\!\!
\begin{dcases*}
\mathcal{N}(\sum_{x'\in R/[l_D,r_D]} \mathbf{e}_{x'}\mathbf{B}^{-1}_{m,D-1} ) & if $l_D=1$ and $r_D=m$\\
\frac{1}{2}[\mathcal{N}(\sum_{x'\in R/[l_D,r_D]} \mathbf{e}_{x'}\mathbf{B}^{-1}_{m,D-1} ), \mathcal{N}(\sum_{x'\in R/[l_D,r_D]} \mathbf{e}_{x'}\mathbf{B}^{-1}_{m,D-1} )] & if $l_D=1$ and $r_D\leq m-1$ \\
\frac{1}{2}[-\mathcal{N}(\sum_{x'\in R/[l_D,r_D]} \mathbf{e}_{x'}\mathbf{B}^{-1}_{m,D-1} ), \mathcal{N}(\sum_{x'\in R/[l_D,r_D]} \mathbf{e}_{x'}\mathbf{B}^{-1}_{m,D-1} )] & otherwise
\end{dcases*}.
\end{align*}

By induction, we can obtain the expression for $\mathcal{N}(\sum_{x\in R}\mathbf{e}_x\mathbf{B}^{-1}_{m,D})$ when $R$ is given.
Notice that if a dimension $d$ is trivial, i.e., $[l_d,r_d]=[1,m]$, then the non-zero entries remain the same during the recursion form dimension $d-1$ to $d$.
For the base case, we have
\begin{equation*}
\mathbf{B}_{m,1}^{-1}=
\frac{1}{2}
\begin{bmatrix}
1 & 0 & \cdots  & 0 & 1\\
-1 & 1 & \ddots & \vdots & 0\\
0 & -1 & \ddots & 0 & \vdots\\
\vdots & \ddots & \ddots & 1 & 0\\
0 & \dots & 0 & -1 & 1\\
\end{bmatrix}
\end{equation*}
Thus, $\mathcal{N}(\sum_{x\in [l_1,r_1]}\mathbf{e}_x\mathbf{B}^{-1}_{m,1})=[1]$, $\frac{1}{2}[1~1]$ or $\frac{1}{2}[-1~1]$.
By induction, we can easily observe that there are $2^{D_R}$ non-zero entries in $\sum_{x\in R}\mathbf{e}_x\mathbf{B}^{-1}_{m,D}$, and each entry is either $2^{-D_R}$ or $-2^{-D_R}$ where $D_R$ is the dimension of $R$.
This is true for the base case, and remains true during the recursive relation above.

Therefore, 
by $\hat{c}(R)=(\frac{e^\eps+1}{e^\eps-1})^{D}(\sum_{x\in R}\mathbf{e}_x\mathbf{B}^{-1}_{m,D})\cdot\ob$, we have
\[
\vrinline{\hat{c}(R)}\leq (\frac{e^\eps+1}{e^\eps-1})^{2D}2^{D_R}(2^{-D_R})^2\max_{x}(\vrinline{\obs_x})
=(\frac{e^\eps+1}{e^\eps-1})^{2D}2^{-D_R}\max_{x}(\vrinline{\obs_x})
\]

Together with the conclusion from the proof of Theorem \ref{thm:mdrq_variance} that $\vrinline{\obs_x} = \bigohinline{(1-(\frac{e^\eps-1}{e^\eps+1})^{2D})n}$, we have
\[
\vrinline{\hat{c}(R)}= \bigoh{(\frac{e^\eps+1}{e^\eps-1})^{2D}2^{-D_R}
(1-(\frac{e^\eps-1}{e^\eps+1})^{2D})n}.
\]
\eop




\subsection{Proof for \cthm\ref{thm:mdrq:weighted}}
\label{app:thm:mdrq:weighted}
When the weights are non-private, as mentioned we can estimate the range query as $\hat{c}_{\bf w}(R)=\sum_{w\in W}w\cdot \hat{c}_{g_w}(R)$. Therefore, the expected squared error is
\begin{equation*}
\begin{aligned}
    \epinline{\| \hat{c}_{\bf w}(R)-c_{\bf w}(R) \|^2} = & \vrinline{\hat{c}_{\bf w}(R)} = \sum_{w \in W} w^2\cdot \vrinline{\hat{c}_{g_w}(R)}
    \\ 
    \leq & \Delta^2 \sum_{w \in W} (\frac{e^\eps+1}{e^\eps-1})^{2D}2^{-D_R} (1-(\frac{e^\eps-1}{e^\eps+1})^{2D}) \cdot |g_w|
    \\ 
    = & \bigoh{\frac{e^\eps+1}{e^\eps-1})^{2D}2^{-D_R} (1-(\frac{e^\eps-1}{e^\eps+1})^{2D})\Delta^2 n}
\end{aligned}
\end{equation*}

When the weights are private, we estimate ${c}_{\bf w}(R)$ as $\hat{c}_{\bf w}(R)=\Delta\cdot \hat{c}(R\times[2,2])$, where the $(D+1)$-th dimension is the weight after randomized rounding (to $\{1, 2\}$). The squared error is
\begin{equation*}
\begin{aligned}
    \epinline{\| \hat{c}_{\bf w}(R)-c_{\bf w}(R) \|^2} 
    = &\vrinline{\hat{c}_{\bf w}(R)} \\
    = &\vrinline{\Delta\cdot \hat{c}(R\times[2,2])} \\
    = &\Delta^2\cdot \vrinline{\hat{c}(R\times[2,2])}
\end{aligned}
\end{equation*}
Now consider $\vrinline{\hat{c}(R\times[2,2])}$. We will follow similar arguments as in Appendix \ref{app:mdrq_range_variance_proof}.
Given a private value $x_i$ and weight $w_i$ held by user $i$, the encoding of the $d$-th dimension of $x_i$ before perturbation is $\mathbf{b}_d\in \{-1,1\}^m$, where the first $(v-1)$ bits are $-1$ and the rest are $1$ if $x_i[d]=v$. 
Let $w_i'\in \{1,2\}$ be the weight value after rounding. The encoding of the $(D+1)$-th dimension is a length-$2$ vector $\mathbf{b}_{D+1}=[1~1]$ if $w_i'=0$ and $\mathbf{b}_{D+1}=[-1~1]$ if $w_i'=1$.
Hence, we have $\prod_{d=1}^{D+1} \mathbf{b}_d[x^+[d]]= (-1)^{S_{>}(x_i,x)}$ for a $(D+1)$-dim value $x^+ = [x~2]$, since $\mathbf{b}_{D+1}[2]=1$. Then after perturbation, we have $\prod_{d=1}^{D+1} \rmatrix_i[d,x^+[d]]= (-1)^{S_{>}(x_i,x)}$ with probability
\[
P_{even}(D+1)=\sum\limits_{i\in[0,D+1], i\text{ is even}}\binom{D+1}{i}p^i(1-p)^{D+1-i}
\] 
where $p=\frac{1}{e^{\eps}+1}$, and $\prod_{d=1}^{D+1} \rmatrix_i[d,x^+[d]] = (-1)\cdot(-1)^{S_{>}(x_i,x)}$ with probability
\[
P_{odd}(D+1)=\sum\limits_{i\in[0,D+1], i\text{ is odd}}\binom{D+1}{i}p^i(1-p)^{D+1-i}.
\]
According to $\obs_x = \sum_{i=1}^n \prod_{d=1}^{D+1} \rmatrix_i[d,x[d]]$, consider the $(D+1)$-dim value $[x~2]$, we have
\[
\begin{aligned}
\vrinline{\obs_{[x~2]}} & = \epinline{\obs_{[x~2]}^2}-\epinline{\obs_{[x~2]}}^2
\\
& = n(P_{odd}(D+1)+P_{even}(D+1))-n(P_{odd}(D+1)-P_{even}(D+1))^2
\\
& = (1-(\frac{e^\eps-1}{e^\eps+1})^{2(D+1)})n.
\end{aligned}
\]

The rest of the proof follows the proof for Theorem \ref{thm:mdrq_range_variance}, by changing dimension $D$ to $D+1$ and $D_R$ to $D_R+1$ due to the private weight dimension.
Thus,
\[
\begin{aligned}
\vrinline{\hat{c}(R\times[2,2])} & = (\frac{e^\eps+1}{e^\eps-1})^{2(D+1)}2^{-(D_R+1)} \cdot \max_{x\in R}(\vrinline{\obs_{[x~2]}})
\\
& = \bigoh{(\frac{e^\eps+1}{e^\eps-1})^{2(D+1)}2^{-(D_R+1)} (1-(\frac{e^\eps-1}{e^\eps+1})^{2(D+1)})n},
\end{aligned}
\]
and
\[
\begin{aligned}
\epinline{\| \hat{c}_{\bf w}(R)-c_{\bf w}(R) \|^2} & = \Delta^2 \cdot \vrinline{\hat{c}(R\times[2,2])}
\\
& = \bigoh{(\frac{e^\eps+1}{e^\eps-1})^{2(D+1)}2^{-(D_R+1)} (1-(\frac{e^\eps-1}{e^\eps+1})^{2(D+1)})\Delta^2 n}.
\end{aligned}
\]
\eop

\subsection{Proof for \clem\ref{lem:qq}}
\label{app:lem:qq}
    Recall that by Equation~\eqref{eq:range_freq_est}, our one-dimensional range query is estimated as $\hat{c}(R)=\hat{c}([1,r])=
    \frac{e^\eps+1}{e^\eps-1}\sum_{x\in [1,r]}\mathbf{e}_x \mathbf{B}^{-1}_{m,1}\cdot \ob= \frac{e^\eps+1}{e^\eps-1}(\obs_1+\obs_r)/2$, where 
    $\obs_{x}= \sum_{i=1}^n \prod_{d=1}^D \mathbf{R}_i[d, x[d]], \forall x\in[m]^D$. Since $\obs_x$ is the summation of $n$ independent random variables $\prod_{d=1}^D \mathbf{R}_i[d, x[d]]\in \{-1,1\}$ 
    , by Chernoff-Hoeffding bound [\cite{boucheron2013concentration}, Theorem 2.8],
    \[
        \pr{|\obs_x-\ep{\obs_x}| \geq t} \leq e^{-\frac{t^2}{2n}}
    \]
    Since $\hat{\sigma}(r)= \frac{e^\eps+1}{e^\eps-1}\cdot\frac{\obs_{1}+\obs_{r}}{2n}$, we have
    \[
        \pr{|\hat{\sigma}(r)-\mathrm{E}(\hat{\sigma}(r))| \geq \frac{e^\eps+1}{e^\eps-1}\cdot\frac{t}{n}} \leq e^{-\frac{t^2}{2n}}
    \]
    Let $\delta=e^{-\frac{t^2}{2n}}$, then $t=\sqrt{2n\log \frac{1}{\delta}}$. Thus, for each frequency estimation for range $[1,r]$, with probability $\geq 1-\delta$, we have $|\hat{\sigma}(r)-\sigma(r)|\leq \frac{e^\eps+1}{e^\eps-1}\cdot \sqrt{\frac{2}{n}\log \frac{1}{\delta}}$.
\eop

\subsection{Proof for \cthm\ref{thm:qq}}
\label{app:thm:qq}
    Let $\Delta=\frac{e^\eps+1}{e^\eps-1}\cdot \sqrt{\frac{2}{n}\log \frac{2\log m}{\delta}}$. We want to show that, w.h.p., $\Err[\hat x^*] \leq 2\Delta$.
    
    By Lemma \ref{lem:qq}, for a particular value $x$, with probability at least $1-\frac{\delta}{2\log m}$, $|\hat{\sigma}(x)-\sigma(x)|\leq \Delta$.
    
    We first want to show that the binary search procedure runs in the {\em desired way} (w.h.p.), which needs to have $\hat\sigma(L) \leq \hat\sigma(R)$ (while $\sigma(L) \leq \sigma(R)$ by definitions).
    As $L$ and $R$ are getting closer to each other, if in any iteration of the algorithm, we have $|\sigma(L) - \sigma(R)| \leq 2\Delta$, then reporting any value between $L$ and $R$ as $\hat x^*$ has the error bounded in the required way.
    Before that happened, an loop invariant of the algorithm is (w.h.p.): 
    $
    \hat\sigma(L) \leq \sigma(L)+\Delta < p \leq \sigma(R)-\Delta \leq \hat\sigma(R),
    $
    and it relies on the events $|\hat\sigma(L) - \sigma(L)| \leq \Delta$ and $|\hat\sigma(R) - \sigma(R)| \leq \Delta$, each of which holds with probability at least $1-\delta/(2\log m)$. Since the binary search procedure touches at most $2\log m$ such $L$'s and $R$'s, it runs in the desired way with probability at least $1-\delta$ (when all these events hold).

    Let $V$ be the multiset of values that are queried by the one-dimensional range query oracle in our mechanism. Since we are using binary search, $|V| \leq 2\log m$.
    For each $x \in V$, with probability at least $1-\frac{\delta}{2\log m}$, $|\hat{\sigma}(x)-\sigma(x)|\leq \Delta$. Thus, with probability at least $1-\delta$, we have $|\hat{\sigma}(x)-\sigma(x)|\leq \Delta$ for all $x \in V$.
    As the aforementioned $L$'s and $R$'s are also in the set $V$, the event that the binary search procedure runs in the desired way holds at the same time. $\hat x^*$ must come from $V$. We can make the following argument, w.h.p.:
    for any $x \in V$ s.t. $\sigma(x) + \Delta < p$, we have $\hat{\sigma}(x) \leq \sigma(x) + \Delta < p$, and thus $x$ will not be chosen by our mechanism; for any $x \in V$ s.t. $p \leq \sigma(x-1) - \Delta$, we have $p \leq \sigma(x-1) - \Delta \leq \hat{\sigma}(x-1)$, and thus $x$ will not be chosen by our mechanism, either.
    %
    %
    %
    Therefore, with probability at least $1-\delta$, the error of the chosen $\hat x^*$ as an estimated $p$-quantile can be bounded as
    \begin{equation*}
    \begin{aligned}
        \Err[\hat{x}^*]=&\inf_{\hat{p} \in I({\hat{x}^*})} |\hat{p} - p|   \\
        \leq &\min(|{\sigma}(\hat{x}^*-1)-p|, |{\sigma}(\hat{x}^*)-p|) \\
        \leq & \min(\Delta, \Delta) \leq 2\Delta.
    \end{aligned}
    \end{equation*}
\eop

\section{Sequential Composability of E-LDP}
\label{app:composability}
Similar to DP and LDP, $E$-LDP also has the following sequential composability.
\begin{proposition}[Composition]\label{prop:composition}
Suppose a randomized algorithm $\ldpalgo_1: \domain \rightarrow \domainoutput_1$ satisfies $E_1$-LDP, and $\ldpalgo_2: \domain \rightarrow \domainoutput_2$ satisfies $E_2$-LDP. If $\ldpalgo_1, \ldpalgo_2$ have independent source of randomness, then the composition of $\ldpalgo_1, \ldpalgo_2$, defined to be $\ldpalgo_3: \domain \rightarrow \domainoutput_1\times\domainoutput_2$ by the mapping $\ldpalgo_3(x)=(\ldpalgo_1(x), \ldpalgo_2(x))$ satisfies $E_3$-LDP, where $E_3(x,x')=E_1(x,x')+E_2(x,x')$ for any pair of $x,x'\in \domain$.
\end{proposition}
\begin{proof}
Consider any $x,x'\in \mathcal{X}$. For any $(y_1,y_2)\in \mathcal{Y}_1\times\mathcal{Y}_2$, we have
\[
\frac{\Pr[\mathcal{A}_3(x)=(y_1,y_2)]}{\Pr[\mathcal{A}_3(x')=(y_1,y_2)]}=
\frac{\Pr[\mathcal{A}_1(x)=y_1]\Pr[\mathcal{A}_2(x)=y_2]}{\Pr[\mathcal{A}_1(x')=y_1]\Pr[\mathcal{A}_2(x')=y_2]}\leq e^{E_1(x,x')}e^{E_2(x,x')}=e^{E_3(x,x')}.
\]
\end{proof}

\end{document}